\documentclass[journal]{IEEEtran}

\usepackage[nocompress]{cite}
\usepackage{graphicx}
\usepackage{amsmath,amssymb,amsfonts}
\usepackage{algorithmic}

\usepackage{booktabs,tabularx,array,makecell,multirow}
\usepackage{pifont} 
\usepackage{threeparttable} 
\newcommand{\cmark}{\ding{51}} 
\newcommand{\xmark}{\ding{55}} 

\usepackage[colorlinks=true, linkcolor=blue, citecolor=blue, urlcolor=blue]{hyperref}

\usepackage{algorithm}
\usepackage{amsthm} 

\newtheorem{lemma}{Lemma}
\newtheorem{proposition}{Proposition}

\usepackage{etoolbox}
\AtBeginEnvironment{figure}{\setlength{\abovecaptionskip}{-5pt}}

\def\E{\mathbb{E}}
\def\R{{\mathbb{R}}}
\def\C{{\mathbb{C}}}
\def\CN{\mathcal{CN}}

\def\cordis{CORDIS}

\def\Nap{N_{ap}}
\def\Nue{N_{ue}}
\def\Nt{N_{t}}
\def\Mt{M_{t}}
\def\Mr{M_{r}}
\def\setA{\mathcal{A}}
\def\setAt{\mathcal{A}_t}
\def\setAr{\mathcal{A}_r}
\def\setU{\mathcal{U}}
\def\setT{\mathcal{T}}
\def\setD{\mathcal{D}}

\def\Pmax{P_\text{max}}
\def\los{LoS}
\def\nlos{NLoS}

\DeclareMathOperator{\tr}{tr}
\DeclareMathOperator{\vect}{vec}
\DeclareMathOperator{\cov}{cov}
\DeclareMathOperator{\mat}{mat}
\DeclareMathOperator{\col}{col}

\begin{document}
\bstctlcite{BSTcontrol}

\title{CORDIS: A Scalable Coordinated Resource Allocation \\ Framework for Distributed Cell-Free ISAC
}


\author{Mehdi~Zafari,~\IEEEmembership{Graduate Student Member,~IEEE,}
        Bj\"{o}rn~Ottersten,~\IEEEmembership{Fellow,~IEEE,}
        \\ and~A.~Lee~Swindlehurst,~\IEEEmembership{Life Fellow,~IEEE}

\thanks{Mehdi Zafari and A. Lee Swindlehurst are with the Department of Electrical Engineering and Computer Science, University of California, Irvine (UCI), CA, USA (e-mail:~\{mzafarid, swindle\}@uci.edu).}

\thanks{Bj\"{o}rn Ottersten is with the University of Luxembourg and the KTH Royal Institute of Technology, Stockholm, Sweden (e-mail: otterste@kth.se).}
}

\maketitle


\begin{abstract}

Integrated Sensing and Communication (ISAC) is envisioned as a key technology for 6G wireless networks, enabling the joint use of spectrum and hardware for sensing and communication.
In multi-static cell-free architectures, coordinating beamforming and power resources across distributed access points (APs) is critical in order to mitigate severe communication-sensing interference.
Most existing ISAC resource allocation solutions rely
on centralized architectures with full network knowledge, which limits their scalability and practicality in distributed cell-free deployments with imperfect channel state information (CSI).
In this paper, we introduce a framework for \textit{COordinated Resource allocation for Distributed ISAC Systems} (\cordis) that optimizes network sensing while suppressing clutter and maintaining per-user communication performance constraints.
Two algorithms are developed: \cordis-Split, a low-overhead scheme that pairs fixed local beamformers with centralized power allocation, and \cordis-ADMM, which jointly optimizes beamforming and power through the consensus Alternating Direction Method of Multipliers (ADMM).
By localizing high-dimensional matrix operations, both algorithms ensure fronthaul overhead and per-AP computation remain independent of antenna and AP counts.
Simulations demonstrate that \cordis-ADMM approaches the centralized performance bound and degrades gracefully under CSI estimation error, remaining effective even with locally rank-deficient channels, while \cordis-Split offers a minimal-overhead alternative.
These results confirm that \cordis\ is a scalable and communication-efficient foundation for robust ISAC in decentralized wireless networks.

\end{abstract}



\section{Introduction}
\label{sec:intro}

\IEEEPARstart{T}{he} sixth generation (6G) of wireless networks is envisioned to transcend the boundaries of traditional connectivity by embedding native perception capabilities into the communication infrastructure.
Integrated Sensing and Communication (ISAC)~\cite{Liu2022dualFunctionalISAC, Wymeersch2021ISAC6G, Liu2026multiDomainOpt}, also known as joint communication and sensing, has emerged as a cornerstone of this vision, officially recognized by the ITU-R IMT-2030 framework~\cite{ITU_R_M2160_2023}, ETSI GR ISC~\cite{ETSI_GR_ISC_001_V111_2025}, and 3GPP standardization efforts~\cite{3GPP_TR_22837_V1940_2024} as a critical usage scenario.
By unifying radar sensing and data transmission onto a shared hardware platform and spectral footprint, ISAC promises to address spectrum scarcity while enabling transformative applications such as autonomous transportation, surveillance, immersive digital twins, and smart manufacturing. However, the deployment of ISAC within legacy cellular architectures faces fundamental physical limitations; rigid cell boundaries and reliance on monostatic sensing geometries often result in severe inter-cell interference and sensing blind spots that compromise the reliability required for safety-critical 6G applications.

In parallel, cell-free massive MIMO has been widely advocated as a post-cellular architecture for 6G, leveraging a dense deployment of distributed access points (APs) that coherently serve users without cell boundaries.
Compared to conventional cellular and small-cell deployments, cell-free networks can provide more uniform user experience through macro-diversity and interference suppression, while naturally supporting scalable implementations via user-centric clustering and local processing \cite{Ngo2017cellfreeVsSmallcell, Interdonato2019ubiquitousCF, Demir2021foundationsCF, Bjornson2020ScalableCFmMIMO}.
From an ISAC perspective, the geographically distributed APs offer an additional advantage: they can act as spatially separated illuminators and/or receivers, enabling multistatic sensing diversity and improved geometric observability, which is particularly attractive for sensing tasks that benefit from multiple spatial views.

Despite these advantages, realizing ISAC over a distributed cell-free architecture raises a central systems challenge: \emph{how to coordinate spatial beamforming and power resources across many APs so as to jointly ensure multi-user communication quality of service (QoS) and sensing performance, while respecting fronthaul capacity and distributed computational limits.}
This coupling is inherently nontrivial because sensing and communication compete for the same spatial degrees of freedom and power budgets, and because coordination must scale with the number of APs without requiring the exchange of high-dimensional channel state information (CSI) or beamforming vectors.
While recent work on scalable cell-free ISAC has made important progress through practical signal processing and sensing architectures, a complementary need is a \emph{system-level resource allocation abstraction} that makes the communication-sensing tradeoffs explicit and enables principled coordination under limited information exchange.
Motivated by this need, we pursue a scalable coordination framework that accommodates both partially centralized and fully decentralized implementations by optimizing spatial beamforming and power across fixed time-frequency resources under practical conditions of imperfect CSI and clutter, leveraging our preliminary developments as a starting point~\cite{Zafari2025confAsilomar}.

\begin{table*}[t]
\centering
\caption{Comparison of Related Works in ISAC Resource Allocation}
\label{tab:related_works}
\renewcommand{\arraystretch}{1.2} 
\begin{tabular}{l c c c c c c c c c c}
\toprule
\textbf{Ref.} & \textbf{Arch.} & \textbf{Channel} & \makecell[c]{\textbf{Sensing}\\\textbf{Metric}} & \textbf{Optimization} &
\makecell[c]{\textbf{Fronthaul}\\\textbf{Load}} & 
\makecell[c]{\textbf{Imp.}\\\textbf{CSI}} &
\makecell[c]{\textbf{Multi-}\\\textbf{static}} &
\textbf{Clutter} &
\makecell[c]{\textbf{Fronth.}\\\textbf{Limits}} &
\makecell[c]{\textbf{Local}\\\textbf{Proc.}} \\
\midrule

\makecell[c]{\cite{Femenias2025ScalableCFmMIMO},\\\cite{Femenias2025scalableISAC_HWI}} &  \makecell[c]{Scalable\\Centralized\\Cell-Free} & \makecell[c]{Correlated\\Rician fading} & \makecell[c]{Detection\\probability\\with MAPRT} & \makecell[c]{Heuristic\\Fractional PA} & \makecell[c]{Global CSI +\\PA weights\\(cmplx scalars)} & \cmark & \cmark & \cmark & \xmark & \xmark\\

\hline

\cite{Mao2024CSregion} & \makecell[c]{Centralized\\Cell-Free} & \makecell[c]{Quasi-static\\blk fading\\Rayleigh/A-LoS} & \makecell[c]{Beampattern\\matching\\MSE} & \makecell[c]{BCD-SCA for\\BF design} & \makecell[c]{Global CSI +\\BF weights\\(cmplx scalars)} & \cmark & \cmark & \xmark & \xmark & \xmark\\

\hline

\cite{Elfiatoure2025multiTargetCF} & \makecell[c]{Distributed\\Cell-Free} & \makecell[c]{Quasi-static\\blk fading\\Rayleigh/A-LoS} & \makecell[c]{Beampattern\\MASR} & \makecell[c]{SCA CVX for\\AP mode and PA} & \makecell[c]{AP mode +\\PA weights\\(real scalars)} & \cmark & \xmark & \xmark & \xmark & \cmark \\

\hline

\cite{Meng2025CoopIsacNet} & \makecell[c]{Cooperative\\Multi-cell\\with CoMP} & \makecell[c]{Uncorrelated\\Rayleigh/A-LoS} & \makecell[c]{Approximate\\CRLB} & \makecell[c]{Cooperative\\cluster size / PA} & \makecell[c]{Global CSI +\\cluster coord.\\(cmplx scalars)} & \xmark & \cmark & \xmark & \cmark & \xmark \\

\hline

\cite{Behdad2024multistatic} & \makecell[c]{Centralized\\Cell-Free\\w/ C-RAN} & \makecell[c]{Blk fading\\corr. Rayleigh\\A-LoS sensing} & \makecell[c]{Detection\\probability\\and SINR} & \makecell[c]{CCP-based PA\\to maximize\\sensing SINR} & \makecell[c]{Global CSI +\\BF/PA weights\\(cmplx scalars)} & \cmark & \cmark & \cmark & \xmark & \xmark \\

\hline

\cite{Huang2022CoordPowerControl} & \makecell[c]{Centralized\\cooperative\\multi-point} & \makecell[c]{Uncorrelated\\Rician fading} & \makecell[c]{CRLB for\\localization} & \makecell[c]{SDR to minimize\\total TX power} & \makecell[c]{Global CSI +\\PA weights\\(cmplx scalars)} & \xmark & \cmark & \xmark & \xmark & \xmark \\

\hline

\cite{Demirhan2025CellFreeISAC} & \makecell[c]{Centralized\\Cell-Free} & \makecell[c]{Blk fading\\with A-LoS +\\point reflector} & \makecell[c]{Sensing SNR} & \makecell[c]{SDP for joint\\comm-sensing\\BF design} & \makecell[c]{Global CSI +\\BF/PA weights\\(cmplx scalars)} & \xmark & \cmark & \xmark & \xmark & \xmark \\

\hline

\cite{Zafari2025confAsilomar} & \makecell[c]{Distributed\\Cell-Free} & \makecell[c]{Uncorrelated\\Ricean fading} & Sensing SNR & \makecell[c]{ADMM for\\BF \&  PA} & \makecell[c]{Interferences +\\consensus\\(real scalars)} & \xmark & \cmark & \xmark & \cmark & \cmark \\

\hline

\textbf{Prop.} & \makecell[c]{\textbf{Distributed}\\\textbf{Cell-Free}} & \makecell[c]{\textbf{Blk fading}\\\textbf{Correlated}\\\textbf{Rician}} & \makecell[c]{\textbf{STAP-based}\\\textbf{SCNR}} & \makecell[c]{\textbf{Consensus}\\\textbf{ADMM}} & \makecell[c]{\textbf{Interferences +}\\\textbf{consensus}\\\textbf{(cmplx scalars)}} & \cmark & \cmark & \cmark & \cmark & \cmark \\
\bottomrule

\multicolumn{11}{@{}p{\textwidth}@{}}{\footnotesize \emph{Note:}
“\cmark” indicates the feature is supported and “\xmark” otherwise.
HWI: Hardware Impairment, BCD: Block Coordinate Decent,  SCA: Successive Convex Approximation, CVX: Convex Optimization, MASR: Mainlobe-to-Average Sensing Ratio, CCP: Concave-Convex Programming, A-LoS: All targets in LoS.
} 

\end{tabular}

\vspace{-2mm}
\end{table*}

\subsection{Background and Related Work}
\label{subsec:related-work}

A first line of work pursues scalable transmit processing and detection.
In \cite{Femenias2025ScalableCFmMIMO}, a scalable cell-free approach is developed 
with user- and target-centric AP clustering, minimum mean-squared error (MMSE) 
precoding, which was later extended to include transceiver hardware
impairments~\cite{Femenias2025scalableISAC_HWI}. Alternatively, \cite{Mao2024CSregion} 
characterizes the communication-sensing (C-S) region
using beampattern-matching and designs a robust beamformer that accounts
for channel estimation and target location uncertainty. These efforts provide
valuable processing and robustness foundations, but both are centralized and
obtain scalability from clustering and centralized processing rather
than from a distributed optimization of beamforming and power under explicit
fronthaul constraints.
 
A second line couples a sensing objective with beamforming and power
optimization. In a C-RAN-assisted cell-free setting, \cite{Behdad2024multistatic} 
combines regularized zero-forcing and a dedicated
sensing beam with a clutter-aware detector and a concave-convex power
allocation that maximizes sensing signal-to-interference-plus-noise ratio (SINR) 
under per-user SINR and per-AP power constraints. A joint multi-static sensing SNR 
and beamforming design with a subsequent power allocation step is developed 
in~\cite{Demirhan2025CellFreeISAC}. In~\cite{Huang2022CoordPowerControl}
minimizes total transmit power under per-user SINR and Cram\'er-Rao bound (CRB) 
constraints via semidefinite relaxation, while \cite{Meng2025CoopIsacNet}
uses stochastic geometry and monotonic optimization to size cooperative clusters
and allocate power under backhaul limits. These designs provide centralized
benchmarks that jointly enforce communication QoS and sensing accuracy, but they
aggregate global CSI (and, in~\cite{Behdad2024multistatic}, sensing data) at the
centralized processing unit (CPU) and do not target distributed, fronthaul-limited operation with local
processing. With the exception of~\cite{Behdad2024multistatic}, they further
assume perfect CSI and neglect clutter.
 
A third line shifts processing toward the network edge. In~\cite{Elfiatoure2025multiTargetCF}, 
distributed APs switch between optimally switch between
communication and sensing modes and optimize power for
spectral efficiency fairness under mainlobe-to-average-sidelobe-ratio (MASR)
beampattern constraints, whereas \cite{Zou2024DistVsCent} shows that
local detection at the receive APs reduces fronthaul signaling
relative to centralized fusion. Although both embrace decentralization and
local processing, the former relies solely on a transmit-beampattern metric and
does not consider target echoes, clutter, or detection, while the latter treats
detection in isolation with pre-determined precoders and idealized, clutter-free
interference, leaving the underlying joint beamforming and power allocation
unaddressed. As summarized in Table~\ref{tab:related_works}, each of these designs meets only
a subset of the capabilities required for practical distributed cell-free ISAC.

\subsection{Research Gap and Contributions}
\label{subsec:contributions}

Practical distributed cell-free ISAC calls for a framework that jointly
optimizes beamforming and power in a distributed way under explicit
fronthaul and local computation limits, using space-time adaptive processing (STAP)-based sensing 
and assuming imperfect CSI and clutter.
We propose \textit{\textbf{CO}ordinated \textbf{R}esource allocation for
\textbf{D}istributed \textbf{I}SAC \textbf{S}ystems (CORDIS)}, a scalable
framework spanning partially centralized and fully decentralized implementations
that significantly generalizes our preliminary designs in~\cite{Zafari2025confAsilomar}. The main
contributions of this work listed below.
 
\begin{itemize}
\item \textbf{System model and analysis.}
We develop a system model for distributed cell-free ISAC over spatially
correlated channels with imperfect CSI and clutter, and derive the 
per-user SINR, MMSE channel estimates, and average
signal-to-clutter-plus-noise ratio (SCNR) that serves as the sensing performance metric 
(Sections~\ref{sec:system-model} and~\ref{sec:dl-isac-tx}).
 
\item \textbf{CORDIS-Split algorithm.}
We propose a low-overhead algorithm that decouples transmit beamforming from power
allocation. Each transmit AP forms a robust local MMSE precoder for communications precoder 
with target nulling, while the CPU solves a convex power
allocation problem, exchanging only a few real scalars per user over the fronthaul
(Section~\ref{sec:cordis-split}).
 
\item \textbf{CORDIS-ADMM algorithm.}
We propose a fully decentralized consensus-ADMM algorithm that jointly optimizes
beamforming and power allocation via a combined augmented-Lagrangian/majorization-%
minimization, with per-AP subproblems and with per-round fronthaul exchanges and 
per-APs computation that are
independent of the array size and number of APs, and that remain effective even
for rank-deficient channels (Section~\ref{sec:cordis-admm}).
 
\item \textbf{Performance evaluation.}
We provide extensive numerical evaluations with
clutter and imperfect CSI, characterizing the communication-sensing trade-off,
clutter suppression, CSI robustness, and low-rank scalability, showing that
CORDIS-ADMM approaches the centralized performance bound while preserving
cell-free scalability (Section~\ref{sec:results})\footnote{All plots, code base, algorithm implementations, and datasets are publicly available at \url{https://github.com/LS-Wireless/CORDIS}.}.
\end{itemize}

\subsection{Notation}
\label{subsec:organization}

The typefaces $x$, $\mathbf{x}$, $\mathbf{X}$, and $\mathcal{X}$ respectively denote scalars, vectors, matrices, and sets. The operators
$\mathbf{X}^*$, $\mathbf{X}^\top$, $\mathbf{X}^H$, and $\mathbf{X}^{-1}$ denote conjugation, transposition, conjugate transposition, and the matrix inverse, respectively. The vectorization, trace, Kronecker product, and Hadamard product operators are denoted by 
$\vect (\cdot)$, $\tr (\cdot)$, $\otimes$, and $\odot$, respectively. Expectation under the distribution of stochastic variable $x$ is denoted by $\E_{x}\{\cdot\}$, while $\R$ and $\C$ denote the sets of real and complex numbers, respectively.
Given a set of scalars $\{\mathrm{s}_n\}_{n=1}^N$ and a set of vectors $\{\mathbf{v}_n\}_{n=1}^N$, the column and matrix operators are defined as $\col\big(\{\mathrm{s}_n\}_{n=1}^N\big) = [\mathrm{s}_1, \cdots, \mathrm{s}_N]^\top$, $\col\big(\{\mathbf{v}_n\}_{n=1}^N\big) = [\mathbf{v}_1^\top, \cdots, \mathbf{v}_N^\top]^\top$, and $\mat\big(\{\mathbf{v}_n\}_{n=1}^N\big) = [\mathbf{v}_1, \cdots, \mathbf{v}_N]$.
Operators $\Re\{\cdot\}$ and $\Im\{\cdot\}$ denote real and imaginary part of a complex scalar, respectively.


\section{System Model}
\label{sec:system-model}

We consider a cell-free ISAC network consisting of $\Nap$ distributed APs, indexed by set $\setA$, that are connected to a CPU through wired fronthaul links in a star topology.
There exist $\Nue$ single-antenna communication users (UEs) indexed by $\setU = \{1, \cdots, \Nue\}$ and $\Nt$ targets indexed by $\setT = \{\Nue+1, \cdots, \Nue+\Nt\}$. The CPU clusters the APs into two non-overlapping subsets $\setAt$ and $\setAr$, assumed to be known.
The APs in $\setAt$ transmit using $\Mt$ antennas to serve users and illuminate targets, while the APs in $\setAr$ use $\Mr$ antennas to detect target echoes. We assume that each transmit AP \(a_t \in \setA_t\) jointly serves all users in \(\setU\) while concurrently illuminating a CPU-selected subset of targets \(\setT_{a_t}\subseteq \setT\), and the CPU assigns every receive AP \(a_r \in \setA_r\) a subset of targets \(\setT_{a_r}\subseteq \setT\) to process for detection or tracking.

We assume a standard block-fading communication channel model that remains constant over a time-frequency coherence block of size $\tau_c$ samples and varies independently from one coherence block to the next~\cite{Femenias2025ScalableCFmMIMO, Mao2024CSregion, Elfiatoure2025multiTargetCF, Behdad2024multistatic, Demirhan2025CellFreeISAC}.
The network operates on time-frequency samples under a synchronous  time division duplex (TDD) protocol.
Each TDD frame is composed of $\tau_f \leq \tau_c$ samples, of which $\tau_p$ samples are used for uplink training, $\tau_u$ for uplink data transmission, and $\tau_d$ for downlink ISAC transmission: $\tau_p + \tau_u + \tau_d = \tau_f$.
During uplink training, UEs transmit pilots to the APs in \(\setAt\), which use them to estimate the corresponding communication channels.
These channel estimates serve as the basis for designing the uplink combiners during the subsequent payload transmission phase,
as well as the downlink precoders used in the ISAC transmission phase.
The signals received at the sensing APs in \(\setAr\) during the downlink phase are processed for target detection and tracking.

\subsection{Communication Channel Model}

The channels between UE $u\in \setU$ and transmit AP $a \in \setAt$ is modeled as
\begin{equation}
    \mathbf{h}_{au} = \sqrt{\beta_{au}} \, (\mathbf{h}_{au}^\text{\los} + \mathbf{h}_{au}^\text{\nlos}),
\end{equation}
where $\beta_{au}$ captures large-scale fading and $\mathbf{h}_{au}^\text{\los}$ and $\mathbf{h}_{au}^\text{\nlos}$ denote the line-of-sight (LoS) and non-LoS (NLoS) components, respectively.
The LoS channel is represented as a deterministic array response with random phase, while the NLoS channel is generated by spatially correlated small-scale fading:
\begin{align}
    \mathbf{h}_{au}^\text{\los} &= \sqrt{\eta_{au}^\text{LoS}} \, e^{j\varphi_{au}} \, \mathbf{a}(\phi_{au}, \theta_{au}), \\
    \mathbf{h}_{au}^\text{\nlos} &= \sqrt{\eta_{au}^\text{NLoS}} \, \left(\mathbf{C}^{1/2}(\phi_{au}, \theta_{au}) \, \mathbf{q}_{au} + \mathbf{B}_{a} \, \mathbf{s}_{au}\right),
\end{align}
where $\eta_{au}^\text{LoS}$ and $\eta_{au}^\text{NLoS}$ characterize the relative LoS/NLoS power gains and satisfy $\eta_{au}^\text{LoS} + \eta_{au}^\text{NLoS} = 1$, $\varphi_{au} \sim \text{Unif}\,[-\pi, \pi]$, and $\mathbf{a}(\phi_{au}, \theta_{au})$ is the array response  parameterized by azimuth and elevation angles $\phi_{au}$ and $\theta_{au}$, respectively. The matrix $\mathbf{C}^{1/2}(\phi_{au}, \theta_{au})$ is a positive semi-definite (PSD) spatial correlation matrix at the $a$-th AP as seen from user $u$, and $\mathbf{q}_{au} \sim \CN(\mathbf{0}, \mathbf{I}_{\Mt})$ accounts for the Rayleigh component.
For simplicity, we define $\mathbf{a}_{au} \triangleq \mathbf{a}(\phi_{au}, \theta_{au})$ and $\mathbf{C}_{au} \triangleq \mathbf{C}(\phi_{au}, \theta_{au})$. The normalized array response satisfies $\|\mathbf{a}_{au}\|^2 = \Mt$. Spatial correlation between the channels of different UEs and a given AP $a$ is modeled by a shared subspace $\mathbf{B}_a \in \C^{\Mt \times r}$, with $r \ll \Mt$.
The corresponding UE-specific scattering coefficients are modeled as $\mathbf{s}_{au} \sim \CN(\mathbf{0}, \boldsymbol{\Sigma}_{au}) \in \C^{r}$.
Hence, the cross-UE channel covariance $\E\{\mathbf{h}_{au} \mathbf{h}_{av}^H \}$ for users $u,v\in \setU$ and $u\neq v$ is
\begin{align}
    \label{eq:comm-channel-covariance}
    \cov(\mathbf{h}_{au}, \mathbf{h}_{av}) &= \sqrt{\beta_{au} \eta_{au}^\text{NLoS} \, \beta_{av} \eta_{av}^\text{NLoS}}\, \mathbf{B}_a \E\{\mathbf{s}_{au} \mathbf{s}_{av}^H \} \mathbf{B}_a^H \nonumber \\
    &\triangleq \xi_{a,uv} \, \mathbf{B}_a \, \boldsymbol{\Gamma}_{a,uv} \, \mathbf{B}_a^H,
\end{align}
where $\boldsymbol{\Gamma}_{a,uv}$ is the cross-user covariance of $\mathbf{s}_{au}$ and  $\mathbf{s}_{av}$, and $\boldsymbol{\Gamma}_{a,uu} = \boldsymbol{\Sigma}_{au}$.
Thus, the spatial correlation matrix characterizing channel vector $\mathbf{h}_{au}$ is obtained as
\begin{align}
    \mathbf{R}_{au} = \beta_{au} \left( \eta_{au}^\text{LoS}\, \mathbf{a}_{au} \mathbf{a}_{au}^H + \eta_{au}^\text{NLoS}\, \widetilde{\mathbf{C}}_{au} \right),
\end{align}
where $\widetilde{\mathbf{C}}_{au} = \mathbf{C}_{au} + \mathbf{B}_a \boldsymbol{\Sigma}_{au} \mathbf{B}_a^H$ is normalized such that $\tr(\widetilde{\mathbf{C}}_{au}) = \Mt$.
Consequently, the total average channel power satisfies $\tr(\mathbf{R}_{au}) = \beta_{au} \Mt$.

\subsection{Sensing and Clutter Channel Model}

During the downlink phase, when a transmit AP $a_t \in \setAt$ and a receive AP $a_r \in \setAr$ collaborate in sensing target $t$ (i.e., $t \in \setT_{a_t}$ and $t \in \setT_{a_r}$), the corresponding sensing channel at slow-time symbol $\tau$ is modeled as the superposition of the target-reflection and an aggregate clutter component (we omit the explicit dependence on $t$ in the following when it is clear from the context):
\begin{equation}
    \mathbf{H}_{a_t a_r}^\text{sens}[\tau] = \mathbf{H}_{a_t a_r}^\text{tgt} + \mathbf{H}_{a_t a_r}^\text{clt}[\tau] \in \C^{\Mr \times \Mt},
\end{equation}
where $\mathbf{H}_{a_t a_r}^\text{tgt}$ represents the rank-one LoS path between $a_t$ and $a_r$ induced by the target reflection, and $\mathbf{H}_{a_t a_r}^\text{clt}[\tau]$ captures the aggregate clutter returns, which exhibit temporal correlation arising from Doppler-dependent scattering.

The LoS component of the sensing channel is modeled as
\begin{align}
    \mathbf{H}_{a_t a_r}^\text{tgt} = s_t \, \sqrt{\beta_{a_t a_r}^\text{tgt}} \, \zeta_t \, \mathbf{a}_{a_r} \mathbf{a}_{a_t}^H \triangleq \alpha_{a_t a_r} \, \mathbf{a}_{a_r} \mathbf{a}_{a_t}^H,
\end{align}
where $\beta_{a_t a_r}^\text{tgt}$ is the bi-static path loss, $\zeta_t \sim \CN(0, \sigma_\text{RCS}^2)$ is the radar cross-section (RCS) of target $t$ as seen by AP pair $(a_t, a_r)$, $\mathbf{a}_{a_r}$ and $\mathbf{a}_{a_t}$ are the receive and transmit array steering vectors, respectively, and $s_t \in \{0, 1\}$ indicates LoS availability: $s_t=0$ when target $t$ is absent or the LoS path is blocked, and $s_t=1$ otherwise.
For brevity we define $\alpha_{a_t a_r} \triangleq s_t \, \sqrt{\beta_{a_t a_r}^{\text{tgt}}}\,\zeta_t$.
The clutter component is modeled using a separable space-time correlation structure~\cite{Richards2010BookRadar, Femenias2025ScalableCFmMIMO}:
\begin{align}
    \label{eq:clutter-channel}
    \mathbf{H}_{a_t a_r}^\text{clt}[\tau] = \sigma_\text{clt}\, \mathbf{C}_{a_r}^{1/2}\, \mathbf{Q}_{a_t a_r}[\tau]\, \mathbf{C}_{a_t}^{1/2},
\end{align}
where $\sigma_\text{clt}$ represents the clutter channel gain, and $\mathbf{C}_{a_r} = \mathbf{C}_{a_r}(\phi_{a_r a_t}, \theta_{a_r a_t})$ and $\mathbf{C}_{a_t} = \mathbf{C}_{a_t}(\phi_{a_t a_r}, \theta_{a_t a_r})$ model the spatial correlation for $a_r$ and $a_t$, respectively. These statistics are assumed to be known and constant over the coherence block, and are normalized to satisfy $\tr(\mathbf{C}_{a_r}) = \Mr$ and $\tr(\mathbf{C}_{a_t}) = \Mt$.
To capture dynamic clutter Doppler characteristics, the elements of the unstructured scattering matrix $\mathbf{Q}_{a_t a_r}[\tau] \in \C^{\Mr \times \Mt}$ are modeled as zero-mean, circularly symmetric complex Gaussian stationary random processes.
Defining $\mathbf{q}_{a_t a_r}[\tau] = \vect(\mathbf{Q}_{a_t a_r}[\tau])$, the temporal auto-correlation of this scattering component is $\E\{\mathbf{q}_{a_t a_r}[\tau_1] \mathbf{q}_{a_t a_r}^H[\tau_2] \} = \rho_\text{clt}(\tau_1 - \tau_2) \mathbf{I}_{\Mr \Mt}$,
where, $\rho_\text{clt}(\Delta\tau)$ is the normalized slow-time correlation function that defines the clutter Doppler power spectral density, with $\rho_\text{clt}(0)=1$.

\begin{proposition}
\label{prop:sensing-channel}
Let $\mathbf{h}_{a_t a_r}^\textnormal{sens}[\tau]$ be the vectorized sensing channel matrix at symbol $\tau$:
\begin{equation}
    \mathbf{h}_{a_t a_r}^\textnormal{sens}[\tau] = \vect(\mathbf{H}_{a_t a_r}^\textnormal{sens}[\tau]) = \mathbf{h}_{a_t a_r}^\textnormal{tgt} + \mathbf{h}_{a_t a_r}^\textnormal{clt}[\tau],
\end{equation}
with vectorized channels $\mathbf{h}_{a_t a_r}^\textnormal{tgt}$ and $\mathbf{h}_{a_t a_r}^\textnormal{clt}$ defined as
\begin{align}
    \mathbf{h}_{a_t a_r}^\textnormal{tgt} &= \alpha_{a_t a_r} \left(\mathbf{a}_{a_t}^* \otimes \mathbf{I}_{\Mr} \right) \mathbf{a}_{a_r}, \\
    \mathbf{h}_{a_t a_r}^\textnormal{clt}[\tau] &= \sigma_\textnormal{clt} \left(\mathbf{C}_{a_t}^{\top/2} \otimes \mathbf{C}_{a_r}^{1/2} \right) \mathbf{q}_{a_t a_r}[\tau],
\end{align}
where $\mathbf{q}_{a_t a_r}[\tau] \sim \CN(\mathbf{0}, \mathbf{I}_{\Mr \Mt})$.
Then, the spatial correlation matrix characterizing $\mathbf{h}_{a_t a_r}^\textnormal{sens}[\tau]$ is expressed as
\begin{equation}
    \mathbf{R}_{a_t a_r}^\textnormal{sens} = \E_{\mathbf{q},\zeta_t}\{\mathbf{h}_{a_t a_r}^\textnormal{sens}[\tau] {\mathbf{h}_{a_t a_r}^\textnormal{sens}}^H[\tau]\} = \mathbf{R}_{a_t a_r}^\textnormal{tgt} + \mathbf{R}_{a_t a_r}^\textnormal{clt},
\end{equation}
with PSD correlation matrices $\mathbf{R}_{a_t a_r}^\textnormal{tgt}$ and $\mathbf{R}_{a_t a_r}^\textnormal{clt}$ defined as
\begin{align}
    \mathbf{R}_{a_t a_r}^\textnormal{tgt} &= s_t \beta_{a_t a_r}^\textnormal{tgt} \sigma_\textnormal{RCS}^2 \left(\mathbf{A}_{a_t}^\top \otimes \mathbf{A}_{a_r} \right), \\
    \mathbf{R}_{a_t a_r}^\textnormal{clt} &= \sigma_\textnormal{clt}^2 \left(\mathbf{C}_{a_t}^\top \otimes \mathbf{C}_{a_r} \right), 
\end{align}
where $\mathbf{A}_{a_t} = \mathbf{a}_{a_t} \mathbf{a}_{a_t}^H$ and $\mathbf{A}_{a_r} = \mathbf{a}_{a_r} \mathbf{a}_{a_r}^H$.
As expected, $\mathbf{R}_{a_t a_r}^\textnormal{clt}$ is strictly independent of the slow-time index $\tau$.
The average channel gain satisfies $\tr(\mathbf{R}_{a_t a_r}^\textnormal{sens}) = (s_t \beta_{a_t a_r}^\textnormal{tgt} \sigma_\textnormal{RCS}^2 + \sigma_\textnormal{clt}^2) \Mr \Mt$.
\end{proposition}

\begin{proof}
The proof is omitted for brevity.
\end{proof}

\subsection{Uplink Channel Estimation}

For the uplink pilot transmission phase, denote the pilot sequence allocated to user $k$ by $\boldsymbol{\varphi}_k \in \C^{\tau_p}$.
In practical scenarios it is often the case that $\Nue \geq \tau_p$, which means that some users must share the same pilot sequence, leading to pilot contamination. Let $\mathcal{P}_k$ denote the set of users sharing the same pilot sequence with user $k$, including itself.
Thus, $\boldsymbol{\varphi}_k^H \boldsymbol{\varphi}_l = \tau_p, \forall l\in \mathcal{P}_k$ and $\boldsymbol{\varphi}_k^H \boldsymbol{\varphi}_l = 0$ otherwise.

The received pilot signal matrix at AP $a$ is expressed as
\begin{align}
    \textstyle \mathbf{Y}_a^\text{p} = \sqrt{P_p} \sum_{k=1}^{\Nue} \mathbf{h}_{ak} \boldsymbol{\varphi}_k^\top + \mathbf{N}_a^\text{p} \; \in \C^{\Mt \times \tau_p},
\end{align}
where $P_p$ is the maximum available transmit power per pilot at the UEs and $\mathbf{N}_a^\text{p}$ denotes noise.
Assuming that transmit APs $a \in \setAt$ have knowledge of the second-order channel statistics (i.e., correlation matrices), they can obtain sufficient statistics for estimating the channel vector $\mathbf{h}_{au}$ by projecting $\mathbf{Y}_a^\text{p}$ on the corresponding pilot sequence vector $\boldsymbol{\varphi}_u$:
\begin{align}
    \label{eq:pilot-observation}
    \textstyle \mathbf{y}_{au}^\text{p} = \mathbf{Y}_a^\text{p} \boldsymbol{\varphi}_u^* = \sqrt{P_p} \tau_p \sum_{k\in \mathcal{P}_u} \mathbf{h}_{ak} + \mathbf{n}_{au}^\text{p},
\end{align}
where $\mathbf{n}_{au}^\text{p} = \mathbf{N}_a^\text{p} \boldsymbol{\varphi}_u^*$.
The resulting vector $\mathbf{y}_{au}^\text{p} \in \C^{\Mt}$ serves as the sufficient statistic for estimating channel vector $\mathbf{h}_{au}$.

\begin{proposition}
\label{prop:channel-estimation}
Under the pilot observation~\eqref{eq:pilot-observation}, the covariance matrix of $\mathbf{y}_{au}^\textnormal{p}$, denoted by $\boldsymbol{\Psi}_{au}$, the cross-covariance matrix $\mathbf{D}_{au} = \cov(\mathbf{h}_{au}, \mathbf{y}_{au}^\textnormal{p})$, and the linear MMSE estimate of the channel vector $\mathbf{h}_{au}$, denoted by $\widehat{\mathbf{h}}_{au}$, are obtained as
\begin{align}
    \mathbf{D}_{au} &= \sqrt{P_p} \tau_p \big(\mathbf{R}_{au} + \sum_{k\in \mathcal{P}_u\setminus u} \xi_{a,uk}\, \mathbf{B}_a \, \boldsymbol{\Gamma}_{a,uk} \, \mathbf{B}_a^H   \big), \\
    \boldsymbol{\Psi}_{au} &= P_p \tau_p^2 \sum_{l\in \mathcal{P}_u} \big( \mathbf{R}_{al} + \sum_{k\in \mathcal{P}_u \setminus l} \xi_{a,lk} \mathbf{B}_a \boldsymbol{\Gamma}_{a,lk} \, \mathbf{B}_a^H \big) \nonumber\\
    & \quad +\, \sigma_{n,a}^2 \tau_p \mathbf{I}_{\Mt}, \\
    \widehat{\mathbf{h}}_{au} &= \mathbf{D}_{au}\, \boldsymbol{\Psi}_{au}^{-1}\, \mathbf{y}_{au}^\textnormal{p}. \label{eq:channel-estimate}
\end{align}
The mean and covariance of the channel estimates are
\begin{align}
    \mathbb{E} \{\widehat{\mathbf{h}}_{au}\} = \mathbf{0}, \
    \widehat{\mathbf{R}}_{au} \triangleq \mathbb{E} \{\widehat{\mathbf{h}}_{au}\widehat{\mathbf{h}}_{au}^H\}
    = \mathbf{D}_{au}\boldsymbol{\Psi}_{au}^{-1}\mathbf{D}_{au}^H.
\end{align}
\end{proposition}

\begin{proof}
The proof is omitted for brevity.
\end{proof}

\section{Downlink ISAC Transmission}
\label{sec:dl-isac-tx}

In each TDD frame of $\tau_f$ symbols, $\tau_d$ symbols are allocated to downlink ISAC transmission.
During this phase, every transmit AP $a_t \in \setAt$ simultaneously delivers payload data to the UEs in $\setU$ and illuminates a CPU-assigned subset of targets $\setT_{a_t}$, while each receive AP $a_r \in \setAr$ collects the corresponding echo measurements for its assigned target subset $\setT_{a_r}$.
We assume that AP clustering as well as target scheduling and AP-target association are determined centrally at the CPU (e.g.,~\cite{Zafari2026ASSENT, Memisoglu2024Scheduling}) and the resulting decisions are communicated to the APs.
In the following we develop mathematical models for downlink multi-UE communication and sensing, leading to expressions for the UE SINR and the target SCNR metrics.

\subsection{Downlink Communication Model}

The downlink signal transmitted by AP $a_t \in \setAt$ in the $\tau$-th symbol period, with $\tau \in \{1, \cdots, \tau_d \}$, is expressed as
\begin{align}
    \mathbf{x}_{a_t}[\tau] &= \sum_{i\in \setD} \mathbf{w}_{a_t i}\, \mathrm{s}_i[\tau] = \sum_{u\in \setU} \mathbf{w}_{a_t u}\, \mathrm{s}_u[\tau] + \sum_{t\in \setT} \mathbf{w}_{a_t t}\, \mathrm{s}_t[\tau] \nonumber \\
    &= \mathbf{W}_{a_t} \mathbf{s}[\tau] = \mathbf{W}_{a_t}^{(c)} \mathbf{s}^{(c)}[\tau] + \mathbf{W}_{a_t}^{(s)} \mathbf{s}^{(s)}[\tau],
\end{align}
where $\setD = \setU \cup \setT = \{1, \cdots, \Nue + \Nt\}$ denotes the index set of all downlink streams, $\mathbf{w}_{a_t u} \in \C^{\Mt}$ is the precoder for the payload data symbol $\mathrm{s}_u [\tau]$ intended for UE $u$, and $\mathbf{w}_{a_t t} \in \C^{\Mt}$ is the precoder for sensing symbol $\mathrm{s}_t [\tau]$ for target $t$. If target $t$ is not assigned to AP $a_t$, i.e., $t \notin \setT_{a_t}$, then $\mathbf{w}_{a_t t} = \mathbf{0}$.
The precoding vectors at AP $a_t$ are stacked as the columns of the precoding matrix $\mathbf{W}_{a_t} = [\mathbf{W}_{a_t}^{(c)}\,|\, \mathbf{W}_{a_t}^{(s)}] \in \C^{\Mt \times (\Nue + \Nt)}$, where $\mathbf{W}_{a_t}^{(c)} = \mat\big(\{\mathbf{w}_{a_t u}\}_{u\in \setU}\big)$ and $\mathbf{W}_{a_t}^{(s)} = \mat\big(\{\mathbf{w}_{a_t t}\}_{t\in \setT}\big)$.
Finally, we define the symbol vectors $\mathbf{s}[\tau] = \col\big(\{\mathrm{s}_i[\tau]\}_{i\in\setD}\big)$, $\mathbf{s}^{(c)}[\tau] = \col\big(\{\mathrm{s}_u[\tau]\}_{u\in\setU}\big)$, and $\mathbf{s}^{(s)}[\tau] = \col\big(\{\mathrm{s}_t[\tau]\}_{t\in\setT}\big)$.

Assuming that both the payload data and sensing symbols are zero-mean i.i.d. complex random variables satisfying $\E\{|\mathrm{s}_u[\tau]|^2 \} = \E\{|\mathrm{s}_t[\tau]|^2 \} = 1$, the signal vector transmitted by AP $a_t$ must satisfy the average power constraint 
\begin{align}
    \E\{\|\mathbf{x}_{a_t}[\tau]\|^2\} &= \E\{\|\mathbf{W}_{a_t}^{(c)}\|_F^2\} + \E\{\|\mathbf{W}_{a_t}^{(s)}\|_F^2\} \nonumber \\
      &= \E\{\|\mathbf{W}_{a_t}\|_F^2\} \leq P_{\max},
\end{align}
where $\Pmax$ is the total available power at each transmit AP. Writing the channel as a sum of its estimate and error, $\mathbf{h}_{a_t u} = \widehat{\mathbf{h}}_{a_t u} + \widetilde{\mathbf{h}}_{a_t u}$, the signal received by UE $u$ for symbol $\tau$ is
\begin{align}
    \mathrm{y}_u[\tau] &= \sum_{a_t \in \setAt} \widehat{\mathbf{h}}_{a_t u}^H \mathbf{x}_{a_t}[\tau] + \sum_{a_t \in \setAt} \widetilde{\mathbf{h}}_{a_t u}^H \mathbf{x}_{a_t}[\tau] + \mathrm{n}_u[\tau] \nonumber \\
    &\triangleq \widehat{\mathrm{y}}_u[\tau] + \widetilde{\mathrm{y}}_u[\tau] + \mathrm{n}_u[\tau],
\end{align}
where $\mathrm{n}_u[\tau] \sim \CN(0, \sigma_{n,u}^2)$ is noise, $\widehat{\mathrm{y}}_u[\tau]$ is the contribution associated with the estimated channel, and $\widetilde{\mathrm{y}}_u[\tau]$ captures the residual distortion due to channel estimation error.
The term $\widehat{\mathrm{y}}_u[\tau]$ can be further expressed as
\begin{align}
    \widehat{\mathrm{y}}_u[\tau] = & \underbrace{\sum_{a_t \in \setAt} \widehat{\mathbf{h}}_{a_t u}^H \mathbf{w}_{a_t u} \mathrm{s}_u[\tau]}_\text{Comm Desired Signal (CDS)} + \underbrace{\sum_{a_t \in \setAt} \sum_{k\in\setU \setminus u} \widehat{\mathbf{h}}_{a_t u}^H \mathbf{w}_{a_t k} \mathrm{s}_k[\tau]}_\text{Multi-User Interference (MUI)} \nonumber \\
    &+ \underbrace{\sum_{a_t \in \setAt} \sum_{t \in \setT} \widehat{\mathbf{h}}_{a_t u}^H \mathbf{w}_{a_t t} \mathrm{s}_t[\tau]}_\text{Sensing-to-Comm Interference (S2CI)}.
\end{align}

Leveraging the independence of the MMSE estimate $\widehat{\mathbf{h}}_{a_t u}$ and the estimation error $\widetilde{\mathbf{h}}_{a_t u}$, the power of $\widetilde{\mathrm{y}}_u[\tau]$ can be obtained as
\begin{align}
    P_\text{CSI-Error} = \E\{|\widetilde{\mathrm{y}}_u[\tau]|^2\} = \sum_{a_t \in \setAt} \tr\big(\mathbf{W}_{a_t}^H \widetilde{\mathbf{R}}_{a_t u} \mathbf{W}_{a_t} \big) \nonumber \\
    = \sum_{a_t \in \setAt} \left(\sum_{k\in\setU} \mathbf{w}_{a_t k}^H \widetilde{\mathbf{R}}_{a_t u} \mathbf{w}_{a_t k} + \sum_{t\in\setT} \mathbf{w}_{a_t t}^H \widetilde{\mathbf{R}}_{a_t u} \mathbf{w}_{a_t t} \right).
\end{align}
Therefore, given the channel estimates and the error covariance matrices $\widetilde{\mathbf{R}}_{a_t u}$, the instantaneous SINR for UE $u$ conditioned on the collective channel estimate matrix $\widehat{\mathbf{H}}$ is 
\begin{align}
    \text{SINR}_u(\widehat{\mathbf{H}}) = \frac{P_\text{CDS}}{P_\text{MUI} + P_\text{S2CI} + P_\text{CSI-Error} + P_\text{Noise}},
\end{align}
with whose terms are specified in~\eqref{eq:sinr} on the next page.

\begin{figure*}[!t]
\begin{equation}
    \label{eq:sinr}
    \text{SINR}_u (\widehat{\mathbf{H}}) = \frac{|\sum_{a_t \in \setAt} \widehat{\mathbf{h}}_{a_t u}^H \mathbf{w}_{a_t u}|^2}{\sum_{k\in\setU \setminus u} |\sum_{a_t \in \setAt}\widehat{\mathbf{h}}_{a_t u}^H \mathbf{w}_{a_t k}|^2 + \sum_{t \in \setT} |\sum_{a_t \in \setAt} \widehat{\mathbf{h}}_{a_t u}^H \mathbf{w}_{a_t t}|^2 + \sum_{a_t \in \setAt} \tr\big(\mathbf{W}_{a_t}^H \widetilde{\mathbf{R}}_{a_t u} \mathbf{W}_{a_t} \big) + \sigma_{n,u}^2}
\end{equation}
\hrulefill
\end{figure*}

\subsection{Multi-Static Sensing Model}

We assume a set \(\setT\) of candidate range-angle bins identified by the network as potential target locations.
Each transmit AP \(a_t \in \setAt\) is assigned to illuminate a subset \(\setT_{a_t}\subseteq\setT\), while each dedicated receive AP \(a_r \in \setAr\) is assigned a subset \(\setT_{a_r}\subseteq\setT\) whose corresponding echoes are processed for detection or tracking.
Since the receive processing at AP $a_r$ follows the same steps for every \(t\in\setT_{a_r}\), we omit the target index $t$ in what follows whenever it is clear from context.
We initially condition the sensing on the transmit signal matrix $\mathbf{X}\triangleq \mat\!\big(\{\mathbf{x}_{a_t}[\tau]\}_{a_t\in\setAt,\;\tau=0,\ldots,T-1}\big)$, assumed known at the sensing receiver.
This yields a conditional SCNR expression.
We then take the expectation over the random transmit symbols to obtain the corresponding average SCNR, used in the subsequent analysis and optimization.

To simplify the presentation, we develop the model for the azimuth-angle-only case, and ignore dependence on the elevation angle $\theta_t$. For the range bin corresponding to target $t \in \setT_{a_r}$, the signal received at AP $a_r \in \setAr$ for time-domain symbol $\tau \in \{1, \cdots, \tau_d\}$ is 
\begin{align}
    \mathbf{y}_{a_r}^\text{(s)}[\tau] = &\sum_{a_t \in \setAt} s_t \sqrt{\beta_{a_t a_r}^\text{tgt}} \zeta_t e^{j\phi_D[\tau]} \mathbf{a}_{a_r}(\phi_t) \mathbf{a}_{a_t}^H \mathbf{x}_{a_t}[\tau] \nonumber \\
    &+ \sum_{a_t \in \setAt} \mathbf{H}_{a_t a_r}^\text{clt}[\tau] \mathbf{x}_{a_t}[\tau] + \mathbf{n}_{a_r}[\tau],
\end{align}
where $\phi_D[\tau]=2\pi f_D T_s \tau$ represents the phase shift due to the Doppler frequency $f_D$, $T_s$ is the symbol period, and $\mathbf{n}_{a_r}[\tau] \sim \CN(\mathbf{0}, \sigma_{n,a_r}^2 \mathbf{I}_{\Mr})$ is noise.
Define $\mathbf{g}_{a_r}[\tau] \triangleq \sum_{a_t \in \setAt} \mathbf{H}_{a_t a_r}^\text{clt}[\tau] \mathbf{x}_{a_t}[\tau] + \mathbf{n}_{a_r}[\tau]$ and $\gamma_{a_r}[\tau] \triangleq \sum_{a_t \in \setAt} \sqrt{\beta_{a_t a_r}^\text{tgt}} \mathbf{a}_{a_t}^H \mathbf{x}_{a_t}[\tau]$, where $\mathbf{g}_{a_r}[\tau]$ denotes the aggregate clutter plus noise term.
Hence, we can rewrite the received signal as $\mathbf{y}_{a_r}^\text{(s)}[\tau] = s_t \zeta_t e^{j\phi_D[\tau]} \gamma_{a_r}[\tau] \mathbf{a}_{a_r}(\phi_t) + \mathbf{g}_{a_r}[\tau]$.

\begin{lemma}
\label{lemma:per-symbol-g-correlation}
For a given transmit signal $\mathbf{X}$, the conditional covariance of the clutter-plus-noise vector $\mathbf{g}_{a_r}[\tau]$ is strictly dependent on the instantaneous transmit signals.
Under the separable space-time clutter model in~\eqref{eq:clutter-channel} with normalized Doppler correlation $\rho_\textnormal{clt}(\Delta\tau)$ and auto-correlation given in Proposition~\ref{prop:sensing-channel}, the conditional cross-time covariance between symbols $\tau_1$ and $\tau_2$ is 
\begin{align}
    &\mathbf{R}_{g_{a_r}}[\tau_1, \tau_2 \,|\, \mathbf{X}] = \sigma_\textnormal{clt}^2\, \rho_\textnormal{clt}(\tau_1-\tau_2) \cdot \nonumber \\
    &\sum_{a_t\in\setAt} \left(\mathbf{x}_{a_t}^H[\tau_2] \mathbf{C}_{a_t} \mathbf{x}_{a_t}[\tau_1] \right) \mathbf{C}_{a_r} + \sigma_{n,a_r}^2 \delta(\tau_1-\tau_2) \mathbf{I}_{\Mr},
\end{align}
where $\delta(\cdot)$ is the Dirac delta function.
The per-symbol conditional covariance is obtained by setting $\tau_1 = \tau_2 = \tau$ and using $\rho_\textnormal{clt}(0) = 1$.
\end{lemma}

\begin{proof}
The proof is omitted for brevity.
\end{proof}

A set of $T$ consecutive snapshots are collected at receive AP $a_r$ and stacked in the matrix $\mathbf{Y}_{a_r}^\text{(s)} = \mat(\{\mathbf{y}_{a_r}^\text{(s)}[\tau]\}_{\tau=0}^{T-1}) \in \C^{\Mr \times T}$. Define $\mathbf{G}_{a_r} = \mat(\{\mathbf{g}_{a_r}[\tau]\}_{\tau=0}^{T-1})$, $\boldsymbol{\gamma}_{a_r} = \vect(\{\gamma_{a_r}[\tau]\}_{\tau=0}^{T-1})$, the Doppler steering vector $\mathbf{b}(f_D) = \vect(\{e^{j\phi_D[\tau]}\}_{\tau=0}^{T-1})$, and $\mathbf{u}_{a_r}(f_D)= \boldsymbol{\gamma}_{a_r} \odot \mathbf{b}(f_D)$.
We can then write the stacked receive matrix as $\mathbf{Y}_{a_r}^\text{(s)} = s_t \zeta_t \mathbf{a}_{a_r}(\phi_t) \mathbf{u}_{a_r}^\top (f_D) + \mathbf{G}_{a_r}$.
The vectorized form of this matrix is
\begin{align}
    \mathbf{y}_\text{ST} \triangleq \vect(\mathbf{Y}_{a_r}^\text{(s)}) = s_t \zeta_t \mathbf{u}_{a_r}(f_D) \otimes \mathbf{a}_{a_r}(\phi_t) + \mathbf{g}_\text{ST},
\end{align}
where $\mathbf{g}_\text{ST} \triangleq \vect(\mathbf{G}_{a_r}) \in \C^{\Mr T}$ collects the aggregate space-time clutter plus noise.

To enable joint space-time adaptive processing (STAP), we define the space-time steering vector parametrized by the hypothesized azimuth angle of arrival $\phi$ and Doppler frequency $f$ as $\mathbf{v}_{a_r}(\phi, f)\triangleq \mathbf{u}_{a_r}(f)\otimes \mathbf{a}_{a_r}(\phi)$.
Given the conditional covariance $\mathbf{R}_{g_\text{ST}|\mathbf{X}} = \E\{\mathbf{g}_\text{ST} \mathbf{g}_\text{ST}^H \,|\, \mathbf{X} \}$, the canonical STAP filter matched to $(\phi,f)$ is obtained as
\begin{align}
    \mathbf{w}_\text{ST}(\phi, f) \triangleq \mathbf{R}_{g_\text{ST}|\mathbf{X}}^{-1} \mathbf{v}_{a_r}(\phi, f). 
\end{align}
Applying the STAP filter to $\mathbf{y}_\text{ST}$ yields the normalized 2D spatial-Doppler spectrum 
\begin{align}
    &\Lambda_{a_r}(\phi, f) = \mathbf{w}_\text{ST}^H(\phi, f)\, \mathbf{y}_\text{ST} \\
    & \ = s_t \zeta_t \mathbf{v}_{a_r}^H(\phi, f) \mathbf{R}_{g_\text{ST}|\mathbf{X}}^{-1} \mathbf{v}_{a_r}(\phi_t, f_D) + \mathbf{v}_{a_r}^H(\phi, f) \mathbf{R}_{g_\text{ST}|\mathbf{X}}^{-1} \mathbf{g}_\text{ST}. \nonumber
\end{align}
The joint maximum likelihood estimates of the target angle and Doppler can be found using
\begin{equation}
    \label{eq:STAP-MLE-estimation}
    (\widehat{\phi}_t, \widehat{f}_D) \triangleq \arg\max_{\phi \in \Phi, f\in\mathcal{F}} \frac{|\Lambda_{a_r}(\phi, f)|^2}{\mathbf{v}_{a_r}^H(\phi, f) \mathbf{R}_{g_\text{ST}|\mathbf{X}}^{-1} \mathbf{v}_{a_r}(\phi, f)} .
\end{equation}
Using the random parameter estimates $(\widehat{\phi}_t, \widehat{f}_D)$ makes derivation of an analytical SCNR expression intractable. Thus, to establish a performance upper bound, we derive the SCNR assuming perfect knowledge of the target's true parameters $(\phi_t, f_D)$. The performance degradation incurred by relying on STAP-derived estimates in~\eqref{eq:STAP-MLE-estimation} will be evaluated in the numerical results.
The conditional SCNR is thus obtained as
\begin{align}
     &\text{SCNR}_{a_r}(\phi_t, f_D \,|\, \mathbf{X)} \nonumber\\
     &= \frac{\E\{|\Lambda_{a_r}(\phi_t, f_D)|^2 \, |\, \mathbf{X}, \mathcal{H}_1 \} - \E\{|\Lambda_{a_r}(\phi_t, f_D)|^2 \, |\, \mathbf{X}, \mathcal{H}_0 \}}{\E\{|\Lambda_{a_r}(\phi_t, f_D)|^2 \, |\, \mathbf{X}, \mathcal{H}_0 \}} \nonumber\\
     &= \sigma_\text{RCS}^2 \mathbf{v}_{a_r}^H(\phi_t, f_D) \mathbf{R}_{g_\text{ST}|\mathbf{X}}^{-1} \mathbf{v}_{a_r}(\phi_t, f_D),
\end{align}
where $\mathcal{H}_1$ represents the hypothesis that an LoS path exists to target (i.e., $s_t = 1$), and $\mathcal{H}_0$ otherwise.
In the following proposition, we derive a tractable expression for the average SCNR, which will be used as the sensing metric in the CORDIS optimization framework.

\begin{proposition}
    \label{prop:expected-SCNR}
    For zero-mean, unit power, spatially and temporally white communication symbols, the space-time clutter-plus-noise covariance matrix $\mathbf{R}_{g_\textnormal{ST}} = \mathbb{E}_{\mathbf{s}}\{\mathbf{R}_{g_\textnormal{ST}|\mathbf{X}}\}$ is block diagonal, with zero slow-time Doppler correlation $\rho_\textnormal{clt}(\Delta\tau)$. The covariance is
    \begin{align}
        \mathbf{R}_{g_\textnormal{ST}} = \mathbf{I}_T \otimes \mathbf{R}_{g_{a_r}},
    \end{align}
    with per-symbol spatial covariance matrix $\mathbf{R}_{g_{a_r}}$ obtained as
    \begin{align}
        \mathbf{R}_{g_{a_r}} &= \E_{\mathbf{s}} \big\{\mathbf{R}_{g_{a_r}}[\tau \,|\, \mathbf{X}] \big\} \nonumber \\
        &= \sigma_\textnormal{clt}^2 \sum_{a_t \in \setAt} \tr\left(\mathbf{W}_{a_t}^H \mathbf{C}_{a_t} \mathbf{W}_{a_t} \right) \mathbf{C}_{a_r} + \sigma_{n,a_r}^2 \mathbf{I}_{\Mr},
    \end{align}
    with $\mathbf{R}_{g_{a_r}}[\tau \,|\, \mathbf{X}]$ given by Lemma~\ref{lemma:per-symbol-g-correlation}.
    Thus, the expected STAP SCNR decouples into the product of temporal processing and spatial processing gains, given by
    \begin{align}
    \label{eq:expedted-scnr}
    \overline{\textnormal{SCNR}}_{a_r} &= \E_{\mathbf{s}} \{\textnormal{SCNR}_{a_r}(\phi_t, f_D | \mathbf{X})\} \\
    &=\sigma_\textnormal{RCS}^2 T \sum_{a_t \in \mathcal{A}_t} \beta_{a_t a_r}^\textnormal{tgt} \left\|\mathbf{a}_{a_t}^H \mathbf{W}_{a_t} \right\|^2 \left( \mathbf{a}_{a_r}^H \mathbf{R}_{g_{a_r}}^{-1} \mathbf{a}_{a_r} \right). \nonumber
    \end{align}
\end{proposition}

\begin{proof}
    The proof is omitted for brevity.
\end{proof}


\section{CORDIS-Split Algorithm}
\label{sec:cordis-split}

The \cordis\ framework aims to develop distributed resource allocation algorithms that support decentralized processing with minimal fronthaul information exchange. In this section we introduce the \textit{\cordis-Split} algorithm, which decouples beamforming (BF) design from power allocation (PA) optimization by allowing each AP to locally construct its beamformers, while the PA variables are optimized centrally at the CPU.
This significantly reduces fronthaul overhead, although its applicability is inherently constrained by the rank and conditioning of the local channel matrix at each AP. To decouple BF and PA, we define a power splitting ratio (PSR) $\rho_{a_t}$ for each AP $a_t$ that determines the communication power $P_{a_t}^{(c)} = \rho_{a_t} \Pmax$ and sensing power by $P_{a_t}^{(s)} = (1-\rho_{a_t}) \Pmax$.
We also define the local downlink channel matrix from AP $a_t$ to all users as $\mathbf{H}_{a_t} = [\mathbf{h}_{a_t 1}, \cdots, \mathbf{h}_{a_t \Nue}]^\top \in \C^{\Nue \times \Mt}$.
Let $\mathbf{y}_\text{DL}[\tau]$ denote the aggregate noise-free downlink signals received by all users at symbol time $\tau$:
\begin{equation}
    \mathbf{y}_\text{DL}[\tau] = \sum_{a_t\in \setAt} \mathbf{H}_{a_t} \mathbf{W}_{a_t} \mathbf{s}[\tau] \triangleq \mathbf{F}\, \mathbf{s}[\tau] \in \C^{\Nue},
\end{equation}
where $\mathbf{F} \triangleq \sum_{a_t\in \setAt} \mathbf{H}_{a_t} \mathbf{W}_{a_t} \in \C^{\Nue \times (\Nue + \Nt)}$.
For full-rank $\mathbf{H}_{a_t}$ and appropriately designed beamformers $\mathbf{W}_{a_t}$, we can ideally approach $\mathbf{F}^* = \left[\mathbf{I}_{\Nue} \, \middle| \, \mathbf{0}\right]$, implying that each user receives only its intended stream without MUI and S2CI.

\subsection{Distributed Robust Beamforming at APs}
\label{subsec:split-distributed-bf}

To enable decentralized processing, we decompose the global ideal transformation $\mathbf{F}^*$ into local reference matrices for each transmit AP,
$\mathbf{F}_{a_t}^* = \alpha_{a_t} \mathbf{F}^*$, where $\alpha_{a_t}$ is a weight assigned based on local channel quality with $\sum_{a_t\in\setAt}\alpha_{a_t}=1$.
To prevent APs with highly correlated or deeply faded channels from exhausting their power budgets, we define $\alpha_{a_t} = \eta(a_t) / \sum_{b_t \in \setAt} \eta(b_t)$, where $\eta(a_t)$ represents the inverse condition number of the estimate $\widehat{\mathbf{H}}_{a_t} = [\widehat{\mathbf{h}}_{a_t 1}, \dots, \widehat{\mathbf{h}}_{a_t \Nue}]^\top$ and $\widehat{\mathbf{h}}_{a_t u}$ are defined in Proposition~\ref{prop:channel-estimation}.
Thus, each AP locally aims to solve the sub-problem $\mathbf{H}_{a_t} \mathbf{W}_{a_t} \approx \mathbf{F}_{a_t}^*$ for $\mathbf{W}_{a_t}$, where $\mathbf{H}_{a_t} = \widehat{\mathbf{H}}_{a_t} + \widetilde{\mathbf{H}}_{a_t}$, and $\widetilde{\mathbf{H}}_{a_t}$ represents the estimation error.

To achieve $\mathbf{H}_{a_t} \mathbf{W}_{a_t}^{(c)} \simeq \alpha_{a_t} \mathbf{I}_{\Nue}$, we define the regularized least-squares problem 
\begin{align}
    \min_{\mathbf{W}_{a_t}^{(c)}} \E_{\widetilde{\mathbf{H}}}\left\{\|(\widehat{\mathbf{H}}_{a_t} + \widetilde{\mathbf{H}}_{a_t}) \mathbf{W}_{a_t}^{(c)} - \alpha_{a_t} \mathbf{I}_{\Nue}\|_F^2 \right\} + \varepsilon \|\mathbf{W}_{a_t}^{(c)}\|_F^2, \nonumber
\end{align}
which provides a robust local MMSE precoder.
By directly incorporating the linear MMSE error covariance $\widetilde{\mathbf{R}}_{\mathbf{H}_{a_t}} = \sum_{u=1}^{\Nue} \widetilde{\mathbf{R}}_{a_t u}$ into the regularized inversion, the local AP shapes its communication beams to proactively suppress transmission in spatial directions characterized by high channel uncertainty.
The resulting unnormalized precoder is given by
\begin{align}
    \label{eq:split-comm-bf}
    \widetilde{\mathbf{W}}_{a_t}^{(c)} = \left( \widehat{\mathbf{H}}_{a_t}^H \widehat{\mathbf{H}}_{a_t} + \widetilde{\mathbf{R}}_{\textbf{H}_{a_t}} + \varepsilon \mathbf{I}_{\Mt} \right)^{-1} \widehat{\mathbf{H}}_{a_t}^H (\alpha_{a_t} \mathbf{I}_{\Nue}),
\end{align}
where $\varepsilon$ is the regularization parameter.
Defining $\widehat{\mathbf{W}}_{a_t}^{(c)} = \widetilde{\mathbf{W}}_{a_t}^{(c)} / \|\widetilde{\mathbf{W}}_{a_t}^{(c)}\|_F$, the communication BF becomes $\mathbf{W}_{a_t}^{(c)} = \sqrt{\rho_{a_t} \Pmax}\, \widehat{\mathbf{W}}_{a_t}^{(c)}$.
We refer to this method as local robust MMSE (LR-MMSE).

For sensing, transmit AP $a_t$ projects the target steering vectors onto the null space of its estimated communication channel matrix using
\begin{align}
    \widehat{\mathbf{P}}_{a_t}^{\perp} = \mathbf{I}_{\Mt} - \widehat{\mathbf{H}}_{a_t}^{H} \left(\widehat{\mathbf{H}}_{a_t} \widehat{\mathbf{H}}_{a_t}^H + \epsilon \mathbf{I}_{\Nue}\right)^{-1} \widehat{\mathbf{H}}_{a_t},
\end{align}
where $\epsilon$ is chosen to ensure invertibility.
For each target $t \in \setT_{a_t}$, the unnormalized projected sensing vector is $\widetilde{\mathbf{w}}_{a_t t} = \widehat{\mathbf{P}}_{a_t}^{\perp} \mathbf{a}_{a_t}(\phi_t)$. Rather than forcing the local AP to heuristically suppress the inevitable S2CI due to CSI estimation errors, which would compromise the sensing gain, \cordis-Split delegates management of the residual S2CI due to CSI estimation errors to the CPU during the power allocation phase.

To illuminate multiple targets without increasing the dimensionality of the centralized optimization, allocation of the sensing power budget among the targets must be resolved locally at each AP. An equal-power baseline allocation ignores both the varying mission-criticality of targets and the heterogeneous spatial conflicts with communication UEs.
To address this, we introduce a priority-aware projection allocation strategy.
Let $\omega_t \geq 0$ denote a system-defined priority weight dictating the required sensing QoS for target $t$.
We design a local target power fraction $\lambda_{a_t t}$ that is proportional to both the target's priority and the squared norm of its null-space projection:
\begin{align}
    \lambda_{a_t t} = \frac{\omega_t \|\widetilde{\mathbf{w}}_{a_t t}\|^2}{\sum_{j\in\mathcal{T}_{a_t}} \omega_j \|\widetilde{\mathbf{w}}_{a_t j}\|^2}.
\end{align}
This strategy acts as a priority-weighted spatial water-filling mechanism.
By scaling with $\|\widetilde{\mathbf{w}}_{a_t t}\|^2$, the AP naturally avoids forcing sensing power toward targets whose spatial signatures are heavily aligned with the estimated communication subspace, inherently minimizing the risk of severe S2CI.
The normalized sensing beamformer for target $t$ is then $\widehat{\mathbf{w}}_{a_t t}^{(s)} = \sqrt{\lambda_{a_t t}}\, \widetilde{\mathbf{w}}_{a_t t} / \|\widetilde{\mathbf{w}}_{a_t t}\|$, yielding the aggregate sensing matrix $\widehat{\mathbf{W}}_{a_t}^{(s)} = [ \widehat{\mathbf{w}}_{a_t 1}^{(s)}, \dots, \widehat{\mathbf{w}}_{a_t |\mathcal{T}_{a_t}|}^{(s)} ]$ such that $\|\widehat{\mathbf{W}}_{a_t}^{(s)}\|_F = 1$.
After applying the centralized PSR, the sensing BF is $\mathbf{W}_{a_t}^{(s)} = \sqrt{(1-\rho_{a_t}) \Pmax} \, \widehat{\mathbf{W}}_{a_t}^{(s)}$.
We refer to this method as null-space conjugate (NS-C) beamforming.

\subsection{Centralized Power Allocation at the CPU}
\label{subsec:split-centralized-pa}

Following local beamformer design, the CPU optimizes the PSRs $\boldsymbol{\rho} = \col(\{\rho_{a_t}\}_{a_t\in\setAt})$ to maximize sensing performance under communication QoS constraints based on the conditional SINR in~\eqref{eq:sinr}:
\begin{align}
    \label{eq:pa-comm-utility}
    &\widehat{\text{SINR}}_u(\widehat{\mathbf{H}}, \widehat{\mathbf{W}}) \triangleq \frac{S_u(\boldsymbol{\rho})}{D_u(\boldsymbol{\rho})} = \nonumber \\
    &= \frac{|\sum_{a_t\in\setAt} \sqrt{\rho_{a_t}} \widehat{\beta}_{a_t u}|^2}{\sum\limits_{a_t\in\setAt} \left[ \rho_{a_t} (\widehat{g}_{a_t u} + e_{a_t u}^{(c)}) + (1-\rho_{a_t})e_{a_t u}^{(s)}\right] + \frac{\sigma_{n,u}^2}{\Pmax}},
\end{align}
where $\widehat{\beta}_{a_t u} = \widehat{\mathbf{h}}_{a_t u}^H \widehat{\mathbf{w}}_{a_t u}^{(c)}$ denotes the effective signal channel, $\widehat{g}_{a_t u} = \sum_{k\in\setU\setminus u} |\widehat{\mathbf{h}}_{a_t u}^H \widehat{\mathbf{w}}_{a_t k}^{(c)}|^2$ the residual MUI,
and $e_{a_t u}^{(c)} = \tr\left(\widehat{\mathbf{W}}_{a_t}^{(c)}{}^H \widetilde{\mathbf{R}}_{a_t u} \widehat{\mathbf{W}}_{a_t}^{(c)} \right)$ and $e_{a_t u}^{(s)} = \tr\left(\widehat{\mathbf{W}}_{a_t}^{(s)}{}^H \widetilde{\mathbf{R}}_{a_t u} \widehat{\mathbf{W}}_{a_t}^{(s)} \right)$ are CSI error terms.
To preserve low fronthaul overhead, we define $\widetilde{g}_{a_t u} = \widehat{g}_{a_t u} + e_{a_t u}^{(c)}$ and require each transmit AP to compute and forward only the three real-valued scalars $\{\widehat{\beta}_{a_t u}, \widetilde{g}_{a_t u}, e_{a_t u}^{(s)}\}$ for each user to the CPU.

\begin{algorithm}[t]
\caption{CORDIS-Split (DistBF-CentPA)}
\label{alg:cordis-split}
\begin{algorithmic}[1]

\REQUIRE
    Channel estimates $\widehat{\mathbf{H}}_{a_t}$, error covariances $\widetilde{\mathbf{R}}_{a_t u}$, steering vectors $\mathbf{a}_{a_t t}$, QoS thresholds $\gamma_u$, penalty $\kappa$.

\STATE \textbf{Initialization:}
    \STATE \quad Send $\eta(a_t)$ to CPU and receive $\sum_{a_t^\prime} \eta(a_t^\prime)$ in return.

\STATE \textbf{Phase I: Distributed BF at each AP $a_t\in\setAt$}
    \STATE \quad Perform LR-MMSE beamforming (\ref{subsec:split-distributed-bf}) for UEs.
    \STATE \quad Perform NS-C beamforming (\ref{subsec:split-distributed-bf}) for target(s).
    \STATE \quad Measure $\{\widehat{\beta}_{a_t u}, \widetilde{g}_{a_t u}, e_{a_t u}^{(s)}\}_{u \in \mathcal{U}}$, $z_{a_t}$, and $\widetilde{q}_{a_t}$.
    \STATE \quad Share these parameters with the CPU.

\STATE \textbf{Phase II: Centralized PA at the CPU}
    \STATE \quad Receive compressed scalars from all transmit APs.
    \STATE \quad Solve problem~\eqref{eq:split-pa-opt} using interior-point methods.~\label{state:solve-pa-opt}
    \STATE \quad Send optimal PSRs $\{\rho_{a_t}^*\}$ to APs $a_t\in \setAt$.
    
\end{algorithmic}
\end{algorithm}

While the expected SCNR is the ultimate metric for target detection, directly maximizing the exact SCNR with respect to $\boldsymbol{\rho}$ yields a highly non-convex fractional programming problem. Instead, we assume the CPU optimizes sensing using a linear surrogate that exploits the separable Kronecker structure of the clutter model, maximizing the priority-weighted target echo power while actively penalizing the total transmit clutter illumination. The echo power for target $t$ at AP $a_r$ is
\begin{multline}
    P_{a_r t}^\text{echo} = T \Pmax \, \sigma_{\text{RCS}, t}^2 \sum_{a_t\in\setAt} \beta_{a_t a_r}^t \cdot \\ 
    \left[\rho_{a_t} \|\mathbf{a}_{a_t}^H(\phi_t) \widehat{\mathbf{W}}_{a_t}^{(c)}\|^2 + (1-\rho_{a_t}) \|\mathbf{a}_{a_t}^H(\phi_t) \widehat{\mathbf{W}}_{a_t}^{(s)}\|^2 \right],
\end{multline}
where $\mathbf{a}_{a_t t} \triangleq \mathbf{a}_{a_t}(\phi_t)$. The sensing utility is defined as
\begin{align}
    U_{cpu}^\text{sens}\left(\boldsymbol{\rho}\right) = \sum_{a_t\in\setAt} \rho_{a_t} \left[z_{a_t} - \kappa \, (q_{a_t}^{(c)} - q_{a_t}^{(s)})\right],
\end{align}
where $\kappa \geq 0$ is a CPU-determined regularization and
\begin{align*}
    z_{a_t} &= \sum_{t, a_r} \omega_t \sigma_{\text{RCS},t}^2 \beta_{a_t a_r}^{t} \left(\|\mathbf{a}_{a_t t}^H \widehat{\mathbf{W}}_{a_t}^{(c)}\|^2 - \|\mathbf{a}_{a_t t}^H \widehat{\mathbf{W}}_{a_t}^{(s)}\|^2 \right) \\
q_{a_t}^{(c)} &= \tr\left( \widehat{\mathbf{W}}_{a_t}^{(c)}{}^H \mathbf{C}_{a_t} \widehat{\mathbf{W}}_{a_t}^{(c)} \right) \; , \;
q_{a_t}^{(s)} = \tr\left( \widehat{\mathbf{W}}_{a_t}^{(s)}{}^H \mathbf{C}_{a_t} \widehat{\mathbf{W}}_{a_t}^{(s)} \right).
\end{align*}
Each AP forwards only the two real-valued scalars $z_{a_t}$ and $\widetilde{q}_{a_t} = (q_{a_t}^{(c)} - q_{a_t}^{(s)})$ to the CPU to form the global objective.

To ensure feasibility under restrictive QoS constraints, we introduce non-negative slack variables  $\boldsymbol{\varepsilon} = \col(\{\varepsilon_u\}_{u\in\setU})$.
Defining $c_{a_t} \triangleq z_{a_t} - \kappa \, \widetilde{q}_{a_t}$ and $\mathbf{c} = \col(\{c_{a_t}\}_{a_t\in\setAt})$, the centralized PA optimization problem is defined as
\begin{subequations}
\label{eq:split-pa-opt}
\begin{align}
    \textbf{P-Split:} \ \max_{\boldsymbol{\rho},\, \boldsymbol{\varepsilon} \geq \mathbf{0}} \quad & \mathbf{c}^\top \boldsymbol{\rho} - \xi \, \mathbf{1}_{\Nue} \boldsymbol{\varepsilon} \label{eq:split-pa-opt-obj} \\
    \text{s.t.} \quad & S_u(\boldsymbol{\rho}) + \varepsilon_u \geq \gamma_u D_u(\boldsymbol{\rho}), \ \forall u \in \mathcal{U}, \label{eq:split-pa-opt-const-sinr} \\ 
    & 0 \leq \rho_{a_t} \leq 1, \ \forall a_t \in \mathcal{A}_t, 
\label{eq:split-pa-opt-const-box}
\end{align}
\end{subequations}
where $\xi \gg 0$ penalizes QoS violations and $\gamma_u$ is the minimum required QoS $\widehat{\text{SINR}}_u$ defined in~\eqref{eq:pa-comm-utility}. The objective \eqref{eq:split-pa-opt-obj} and bounds \eqref{eq:split-pa-opt-const-box} are affine, and $S_u(\boldsymbol{\rho})$ is strictly concave. Thus, \eqref{eq:split-pa-opt-const-sinr} lower-bounds a concave function by an affine one, defining a convex feasible set that is solvable via standard interior-point methods. The \cordis-Split implementation is summarized in Algorithm~\ref{alg:cordis-split}.


\section{CORDIS-ADMM Algorithm}
\label{sec:cordis-admm}

\cordis-Split minimizes fronthaul signaling through spatial decoupling, which inherently limits its performance particularly when the local AP channels are rank-deficient or ill-conditioned.
To overcome this limitations, we introduce the \textit{\cordis-ADMM} algorithm in this section.
In contrast to \cordis-Split, \cordis-ADMM jointly optimizes beamforming and power allocation through an iterative consensus ADMM procedure~\cite{Boyd2011Book, Zafari2024ADMM, Zafari2025confAsilomar}, allowing the APs to coordinate their local decisions while relying only on local imperfect CSI and limited information exchange.
This design is particularly attractive when the number of UEs exceeds the spatial degrees of freedom available at individual APs, or when stronger coordination is required to balance communication and sensing objectives. The resulting algorithm preserves the distributed spirit of the \cordis\ framework while providing a more powerful mechanism for joint optimization. 

\subsection{Global Optimization Problem}
\label{subsec:admm-global-opt}

The ultimate objective of the cell-free ISAC network can be formulated as 
\begin{subequations}
\label{eq:admm-p-global}
\begin{align}
    \textbf{P-Global:} \ \max_{\{\mathbf{W}_{a_t}\}} \quad & \sum_{a_r\in\setAr} \sum_{t\in\setT_{a_r}} \omega_t \,\overline{\text{SCNR}}_{a_r,t} \label{eq:admm-p-global-obj}\\
    \text{s.t.} \quad & \text{SINR}_u(\widehat{\mathbf{H}}) \geq \gamma_u, \ \forall u\in\setU, \label{eq:admm-p-global-sinr}\\
    & \|\mathbf{W}_{a_t}\|_F^2 \leq \Pmax, \ \forall a_t\in\setAt, \label{eq:admm-p-global-power}
\end{align}
\end{subequations}
where $\omega_t$ is the target priority weight, $\overline{\text{SCNR}}_{a_r,t}$ is given by~\eqref{eq:expedted-scnr}, and $\text{SINR}_u(\widehat{\mathbf{H}})$ by~\eqref{eq:sinr}.
Directly optimizing \eqref{eq:admm-p-global-obj} in a distributed framework is challenging since the SCNR denominator relies on $\mathbf{R}_{g_{a_r}}^{-1}$, which is a coupled function of the total clutter illumination generated by all APs.
Consequently, \eqref{eq:admm-p-global-obj} is a highly non-convex sum-of-fractions problem.
While such fractional programs can in principle be decoupled using advanced optimization techniques, applying them to a cell-free network requires iteratively freezing the coupled denominator via auxiliary variables, which introduces a computationally prohibitive triple-loop algorithm and induces gradient blindness, severing the explicit clutter penalty.

As with \cordis-Split, to guarantee stability and preserve the low fronthaul overhead, 
we introduce a clutter-aware linear surrogate for sensing defined as the aggregate priority-weighted target echo power minus a regularized penalty for the total transmit clutter illumination:
\begin{align}
    \label{eq:admm-linear-objective}
    &U_{cpu}^\text{sens}\left(\mathcal{W}\right) = \sum_{a_r\in\setAr} \sum_{t\in\setT_{a_r}} \omega_t P_{a_r t}^\text{echo} - \kappa P^\text{clutter} \\
    & \quad = \sum_{a_t\in\setAt} \left[\sum_{t\in\setT} \overline{\omega}_t \overline{\beta}_{a_t}^t \|\mathbf{a}_{a_t t}^H \mathbf{W}_{a_t}\|^2 - \kappa \tr\left(\mathbf{W}_{a_t}^H \mathbf{C}_{a_t} \mathbf{W}_{a_t}\right)\right], \nonumber
\end{align}
where $\mathcal{W} = \{\mathbf{W}_{a_t}\}_{a_t\in\setAt}$ denotes the set of all transmit beamformers, $\kappa$ is the global clutter penalty known to all APs, $\overline{\omega}_t = \omega_t T \sigma_{\text{RCS},t}^2$, and $\overline{\beta}_{a_t}^t = \sum_{a_r\in\setAr} \beta_{a_t a_r}^t$.
Defining $U_{cpu}^\text{sens}\left(\mathcal{W}\right) \triangleq \sum_{a_t\in\setAt} U_{a_t}^\text{sens}(\mathbf{W}_{a_t})$, the local utility $U_{a_t}^\text{sens}(\mathbf{W}_{a_t})$ can be optimized at each AP $a_t$.

\subsection{SCA-SOCP Convex Relaxation}
\label{subsec:admm-sca-socp}

Maximizing the convex target echo power $\|\mathbf{a}_{a_t t}^H \mathbf{W}_{a_t}\|^2$ in~\eqref{eq:admm-linear-objective} is a non-convex problem, so
we apply successive convex approximation (SCA) using a first-order Taylor expansion around the previous iteration's beamforming matrix $\mathbf{W}_{a_t}^{(n)}$ to define the affine lower-bound 
\begin{multline}
    \|\mathbf{a}_{a_t t}^H \mathbf{W}_{a_t}\|^2 \geq 2\, \Re\left\{\tr\left(\big(\mathbf{W}_{a_t}^{(n)}\big)^H\mathbf{a}_{a_t t} \mathbf{a}_{a_t t}^H \mathbf{W}_{a_t}\right)\right\} \\ - \|\mathbf{a}_{a_t t}^H \mathbf{W}_{a_t}^{(n)}\|^2.
\end{multline}
Substituting this into~\eqref{eq:admm-linear-objective} and removing constants, the modified local sensing utility is 
\begin{align}
    \label{eq:admm-local-sensing-utility}
    &\widetilde{U}_{a_t}^\text{sens}\big(\mathbf{W}_{a_t}^{(n)}; \mathbf{W}_{a_t}\big) =  \sum_{t\in\setT} \overline{\omega}_t \overline{\beta}_{a_t}^t \cdot \\ 
    & 2 \Re\left\{\tr\left(\big(\mathbf{W}_{a_t}^{(n)}\big)^H\mathbf{a}_{a_t t} \mathbf{a}_{a_t t}^H \mathbf{W}_{a_t}\right)\right\} - \kappa \tr\left(\mathbf{W}_{a_t}^H \mathbf{C}_{a_t} \mathbf{W}_{a_t}\right), \nonumber
\end{align}
which is a concave function of $\mathbf{W}_{a_t}$ and can be maximized.

The remaining challenge lies in the QoS constraint in~\eqref{eq:admm-p-global-sinr}, which 
is inherently non-convex due to its fractional form and the coupling of quadratic interference terms. To address this, we first
take the square root of both sides
\begin{align}
    \textstyle \sqrt{P_\text{Interf} + \sigma_{n,u}^2} \leq \frac{1}{\sqrt{\gamma_u}} \left|\sum_{a_t\in\setAt} \widehat{\mathbf{h}}_{a_t u} \mathbf{w}_{a_t u}\right|,
\end{align}
where $P_\text{Interf} = P_\text{MUI} + P_\text{S2CI} + P_\text{CSI-Error}$ is the total interference power.
Next, we aggregate the MUI terms $\sum_{a_t} \widehat{\mathbf{h}}_{a_t u}^H \mathbf{w}_{a_t k}$ for all $k\in\setU\setminus u$, the S2CI terms $\sum_{a_t} \widehat{\mathbf{h}}_{a_t u}^H \mathbf{w}_{a_t t}$ for all $t\in\setT$, the CSI error vector $\vect(\widetilde{\mathbf{R}}_{a_t u}^{1/2} \mathbf{W}_{a_t})$ for all $a_t\in\setAt$, and the noise power $\sigma_{n,u}$ into a single stacked vector $\mathbf{v}_u\big(\{\mathbf{W}_{a_t}\}\big)$.
Since the square root of the interference power equals the $\ell_2$-norm $\|\mathbf{v}_u\|_2$, we transform~\eqref{eq:admm-p-global-sinr} into a standard, strictly convex Second-Order Cone (SOC) constraint:
\begin{equation}
\label{eq:socp_constraint}
    \textstyle \|\mathbf{v}_u({\mathbf{W}_{a_t}})\|_2 \le \frac{1}{\sqrt{\gamma_u}} \Re\left\{ \sum_{a_t\in\mathcal{A}_t} \widehat{\mathbf{h}}_{a_t u}^H \mathbf{w}_{a_t u} \right\}.
\end{equation}
Although convex, the variables in \eqref{eq:socp_constraint} remain inextricably coupled across all transmit APs $a_t \in \mathcal{A}_t$.
Satisfying this constraint locally is impossible, necessitating the introduction of a decentralized consensus approach.

\subsection{Local Problems and Consensus ADMM}
\label{subsec:admm-consensus-admm}

Define the local contribution vector for AP $a_t$ and user $u$
\begin{align}
    \label{eq:admm-local-contribution}
    \mathbf{l}_{a_t u}(\mathbf{W}_{a_t}) = \big[x_{a_t u}, (\mathbf{i}_{a_t u}^{(c)})^\top, (\mathbf{i}_{a_t u}^{(s)})^\top, e_{a_t u}\big]^\top.
\end{align}
that contains the local desired signal gain $x_{a_t u} = \widehat{\mathbf{h}}_{a_t u}^H \mathbf{w}_{a_t u}$, local MUI vector $\mathbf{i}_{a_t u}^{(c)} = \col(\{\widehat{\mathbf{h}}_{a_t u}^H \mathbf{w}_{a_t k}\}_{k\neq u})$, local S2CI vector $\mathbf{i}_{a_t u}^{(s)} = \col(\{\widehat{\mathbf{h}}_{a_t u}^H \mathbf{w}_{a_t t}\}_{t\in\setT})$, and local CSI error bound $e_{a_t u} = \|\widetilde{\mathbf{R}}_{a_t u}^{1/2} \mathbf{W}_{a_t}\|_F^2$.
To enforce the globally coupled SOC constraint at the CPU, we introduce a global consensus vector $\mathbf{z}_u$ that encapsulates the signals that arrive at UE $u$, 
\begin{align}
    \mathbf{z}_u = \big[z_u^\text{CDS}, (\mathbf{z}_u^\text{MUI})^\top, (\mathbf{z}_u^\text{S2CI})^\top, z_u^\text{err}\big]^\top,
\end{align}
subject to the consensus requirement $\mathbf{z}_u = \sum_{a_t\in\setAt} \mathbf{l}_{a_t u}$.
The augmented Lagrangian for the global problem incorporates the separable SCA sensing objectives and explicitly penalizes the $\ell_2$-norm difference between the aggregate local AP contributions and the global consensus variables:
\begin{multline}
    \label{eq:admm-augmented-Lagrangian}
    \textstyle \mathcal{L}_\rho \big(\{\mathbf{W}_{a_t}\}, \{\mathbf{z}_u\}, \{\boldsymbol{\nu}_u\}\big) = \sum_{a_t\in\setAt} \widetilde{U}_{a_t}^\text{sens}\big(\mathbf{W}_{a_t}^{(n)}; \mathbf{W}_{a_t}\big) \\
    \textstyle - \frac{\rho}{2} \sum_{u\in\setU} \left\|\sum_{a_t\in\setAt} \mathbf{l}_{a_t u} (\mathbf{W}_{a_t}) - \mathbf{z}_u + \boldsymbol{\nu}_u\right\|_2^2, 
\end{multline}
where $\boldsymbol{\nu}_u$ is the dual variable and $\rho > 0$ is the standard ADMM penalty parameter dictating the convergence step size.

During the update step, each transmit AP optimizes its precoding matrix $\mathbf{W}_{a_t}$ to maximize its portion of the augmented Lagrangian, assuming the variables from other APs ($b_t \neq a_t$) and the CPU are fixed.
Let $\mathbf{\Sigma}_{a_t u}^{(n)} \triangleq \sum_{b_t \neq a_t} \mathbf{l}_{b_t u}^{(n)} - \mathbf{z}_u^{(n)} + \boldsymbol{\nu}_u^{(n)}$ denote the residual at iteration $n$ passed to AP $a_t$.
The local subproblem evaluated at transmit AP $a_t$ is
\begin{subequations}
\label{eq:admm-p-local}
\begin{align}
    \textbf{P-Local:}  \label{eq:admm-p-local-obj} \\ \max_{\mathbf{W}_{a_t}} \ \ &  \textstyle \widetilde{U}_{a_t}^\text{sens}\big(\mathbf{W}_{a_t}^{(n)}; \mathbf{W}_{a_t}\big)
     - \frac{\rho}{2} \sum\limits_{u \in \mathcal{U}} \left\| \mathbf{l}_{a_t, u}(\mathbf{W}_{a_t}) + \mathbf{\Sigma}_{a_t, u}^{(n)} \right\|_2^2   \nonumber \\
    \text{s.t.} \ \ & \|\mathbf{W}_{a_t}\|_F^2 \le P_{\max}. \label{eq:admm-p-local-const}
\end{align}
\end{subequations}
Since the objective \eqref{eq:admm-p-local-obj} combines a linear-quadratic (SCA) sensing utility, defined in~\eqref{eq:admm-local-sensing-utility}, with a negative quadratic ADMM penalty, the function is strictly concave.
Coupled with \eqref{eq:admm-p-local-const}, the local problem~\eqref{eq:admm-p-local} is a Quadratically Constrained Quadratic Program (QCQP), which can be efficiently solved using standard local interior-point methods.

After receiving the updated local contribution vectors $\mathbf{l}_{a_t u}^{(n+1)}$ from all APs, the CPU updates the global consensus vector $\mathbf{z}_u$ for each UE.
To represent the SOC constraint, define
\begin{align}
    \text{SOC}(\mathbf{z}_u) = \sqrt{\|\mathbf{z}_u^\text{MUI}\|_2^2 + \|\mathbf{z}_u^\text{S2CI}\|_2^2 + z_u^\text{err} + \sigma_{n,u}^2} - \frac{\Re\{z_u^\text{CDS}\}}{\sqrt{\gamma_u}}. \nonumber
\end{align}
Since the global consensus variables $\mathbf{z}_u$ and their SOC constraints are decoupled across UEs, the global update decomposes into $\Nue$ independent SOCP projections. Thus, the central consensus problem at the CPU is formulated as
\begin{subequations}
\label{eq:admm-p-central}
\begin{align}
    \textbf{P-Central:} \ \min_{\mathbf{z}_{u}, \varepsilon_u} \ & \textstyle \left\|\sum_{a_t} \mathbf{l}_{a_t u}^{(n+1)} + \boldsymbol{\nu}_u^{(n)} - \mathbf{z}_u\right\|_2^2 + \xi \varepsilon_u \label{eq:admm-p-central-obj}\\
    \text{s.t.} \ & \text{SOC}(\mathbf{z}_u) \leq \varepsilon_u, \label{eq:admm-p-central-const-1}\\
    & \Im\big\{z_u^\text{CDS}\big\} = 0, \ z_u^\text{err} \geq 0, \label{eq:admm-p-central-const-2}
\end{align}
\end{subequations}
where $\varepsilon_u > 0$ is a slack variable penalized by $\xi \gg 0$ and used to guarantee algorithmic stability.
This problem is a simple, low-dimensional projection of $\mathbf{z}_u$ onto a cone defined by \eqref{eq:admm-p-central-const-1}-\eqref{eq:admm-p-central-const-2}.
The CPU computes this projection for each UE in parallel and updates the dual multiplier via standard dual ascent.
The complete procedure of the \cordis-ADMM algorithm is summarized in Algorithm~\ref{alg:cordis-admm}.
For the stopping criterion, the primal and dual residuals are evaluated as
\begin{align}
    \textstyle r_\text{pri}^{(n+1)} & \textstyle = \sum_{u\in\setU} \|\sum_{a_t} \mathbf{l}_{a_t u}^{(n+1)} - \mathbf{z}_u^{(n+1)}\|_2 \leq \epsilon_\text{pri}, \label{eq:admm-primal-residual} \\
    \textstyle r_\text{dual}^{(n+1)} & \textstyle = \sum_{u\in\setU} \rho \, \|\mathbf{z}_u^{(n+1)} - \mathbf{z}_u^{(n)}\|_2 \leq \epsilon_\text{dual} \label{eq:admm-dual-residual}.
\end{align}

\begin{algorithm}[t]
\caption{\cordis-ADMM (Joint BF-PA)}
\label{alg:cordis-admm}
\begin{algorithmic}[1]

\REQUIRE
    Channel estimates $\widehat{\mathbf{H}}_{a_t}$, error covariances $\widetilde{\mathbf{R}}_{a_t u}$, steering vectors $\mathbf{a}_{a_t t}$, clutter covariances $\mathbf{C}_{a_t}$, QoS thresholds $\gamma_u$, $\Pmax$, ADMM penalty $\rho$, max iterations $N_{\max}$, and tolerances $\epsilon_\text{pri}$ and $\epsilon_\text{dual}$.

\STATE \textbf{Initialization:}
    \STATE \quad Execute Phase~I of \cordis-Split to obtain $\{\mathbf{W}_{a_t}^{(0)}\}$.
    \STATE \quad Initialize $\mathbf{z}_u^{(0)} = \sum_{a_t} \mathbf{l}_{a_t u}^{(0)}$ and  $\boldsymbol{\nu}_u^{(0)} = \mathbf{0}$ for all $u \in \mathcal{U}$.

    \STATE \quad CS broadcasts initial $\widetilde{\boldsymbol{\Sigma}}_u^{(0)}$ to all APs $a_t\in\setAt$.

\STATE \textbf{Phase I: Parallel Local AP Updates for all $a_t\in\setAt$}~\label{state:admm_start}
    \STATE \quad Reconstruct local residual: $\boldsymbol{\Sigma}_{a_t u}^{(n)} \leftarrow \widetilde{\boldsymbol{\Sigma}}_u^{(n)} - \mathbf{l}_{a_t u}^{(n)}$.
    \STATE \quad Solve the QCQP problem in \eqref{eq:admm-p-local}.
    \STATE \quad Compute $\mathbf{l}_{a_t u}^{(n+1)} \ \forall u\in\setU$ as in~\eqref{eq:admm-local-contribution} and send to CPU.
\STATE \textbf{Phase II: Centralized Global Updates at the CPU}
    \STATE \quad Solve problem~\eqref{eq:admm-p-central} to obtain $\mathbf{z}_u^{(n+1)}$ for all $u\in\setU$.
    \STATE \quad Update dual $\boldsymbol{\nu}_u^{(n+1)} \leftarrow \boldsymbol{\nu}_u^{(n)} + \sum_{a_t} \mathbf{l}_{a_t u}^{(n+1)} - \mathbf{z}_u^{(n+1)}$.
    \STATE \quad Compute $\widetilde{\boldsymbol{\Sigma}}_u^{(n+1)} \leftarrow \sum_{a_t} \mathbf{l}_{a_t u}^{(n+1)} - \mathbf{z}_u^{(n+1)} + \boldsymbol{\nu}_u^{(n+1)}$.
    \STATE \quad Broadcast $\widetilde{\boldsymbol{\Sigma}}_u^{(n+1)}$ to all transmit APs $a_t\in\setAt$.

\STATE \textbf{Check Convergence Using Primal Dual Residual:}
    \STATE \quad If [\eqref{eq:admm-primal-residual} and \eqref{eq:admm-dual-residual}] or [$n=N_{\max}$]: Exit.
    \STATE \quad Else: Set $n \leftarrow n+1$ and return to step \ref{state:admm_start}.

\end{algorithmic}
\end{algorithm}

 
\textit{Algorithmic structure (ALM/MM).}
\cordis-ADMM is a nested augmented-Lagrangian/majorization-minimization (ALM/MM)
scheme. The MM layer acts on the linear sensing
surrogate~\eqref{eq:admm-linear-objective} rather than the fractional SCNR. The
first-order expansion~\eqref{eq:admm-local-sensing-utility} minorizes the convex
echo power and is tight at $\mathbf{W}_{a_t}^{(n)}$, so maximizing the concave
surrogate monotonically improves the true utility, with the expansion point
refreshed each round. The ALM layer enforces the coupled SOC QoS
constraint~\eqref{eq:socp_constraint} via the augmented
Lagrangian~\eqref{eq:admm-augmented-Lagrangian} and the slack penalty
$\xi\varepsilon_u$, split into the per-AP QCQP~\eqref{eq:admm-p-local} and the
central consensus projection~\eqref{eq:admm-p-central} and solved by consensus
ADMM with dual ascent on $\boldsymbol{\nu}_u$. The two layers run in a single
outer loop that re-linearizes the surrogate and refreshes the consensus each
round, while the only inner loop is the per-user SOC
projection~\eqref{eq:admm-p-central}. Because $\mathbf{l}_{a_t u}$
in~\eqref{eq:admm-local-contribution} contains quantities of heterogeneous scale,
we normalize it so that the residuals
in~\eqref{eq:admm-primal-residual},~\eqref{eq:admm-dual-residual} are
dimensionless and a single threshold $\epsilon_\text{pri}=\epsilon_\text{dual}=1$
applies. Warm-starting from \cordis-Split 
keeps the round count $T_\text{ADMM}$ moderate.

 
\textit{On adaptive penalty selection.}
Residual balancing, which adapts $\rho$ to keep the primal and dual residuals
comparable~\cite{Boyd2011Book}, presumes an unconstrained quadratic consensus
update whose residual ratio varies monotonically with $\rho$. Here the central
update~\eqref{eq:admm-p-central} is instead a constrained per-user SOC projection
whose active set, and hence $\rho$-sensitivity, changes nonlinearly across
rounds, so adapting $\rho$ provides no gain in efficiencey 
Since the per-round contribution vector is exchanged every iteration,
the cumulative fronthaul scales with $T_\text{ADMM}$. Thus, we fix $\rho$ and
control $T_\text{ADMM}$ via a \cordis-Split warm start. By systematically decoupling the constraints, 
\cordis-ADMM reduces the network scaling from $\mathcal{O}(|\mathcal{A}_t|^3)$ to $\mathcal{O}(1)$, 
enabling scalable robust precoding in dense cell-free deployments.

\section{Simulation Results}
\label{sec:results}

In this section, we numerically evaluate the proposed \cordis\ framework and
quantify the communication-sensing trade-offs it exposes. 
We consider a distributed cell-free ISAC network in which $\Nap$ transmit
APs are uniformly deployed on a circle of radius $650$\,m, while $\Nue$
single-antenna users and $\Nt$ targets are placed uniformly at
random within a concentric disc of radius $1$\,km. For each target, the
closest AP is designated as the sensing receive AP
$a_r\in\setAr$. Each AP is equipped with a uniform circular array of
$\Mt=\Mr=M$ half-wavelength spaced antennas operating at a carrier frequency of
$3.5$\,GHz. The communication links follow the spatially correlated Rician
fading model of Section~\ref{sec:system-model}, instantiated as a 3GPP
urban-microcell (UMi) street-canyon channel using the path-loss, LoS 
probability, shadow-fading, and angular-spread parameters of
3GPP TR~38.901~\cite{3GPP_TR_22837_V1940_2024} (mean LoS Rician $K$-factor $9$\,dB; shadow
fading $4$\,dB and $7.82$\,dB for LoS and NLoS, respectively), and the channel
estimates $\widehat{\mathbf{h}}_{a_t u}$ are obtained with the linear MMSE
estimator of Proposition~\ref{prop:channel-estimation}. For a system bandwidth
$B=20$\,MHz with a receiver noise figure of $7$\,dB at $T_0=290$\,K, the thermal
noise floor is $\sigma_n^2=k_B T_0 B\,\mathrm{NF}\approx4.0\times10^{-13}$\,W
($-94$\,dBm), and thus the per-AP transmit-power budget is
$\Pmax\approx20$\,W ($43$\,dBm, i.e., $\Pmax/\sigma_n^2=137$\,dB) and the nominal
uplink pilot power is $P_p\approx2.5$\,W ($34$\,dBm, i.e.,
$P_p/\sigma_n^2=128$\,dB). Unless stated otherwise, the parameters in
Table~\ref{tab:sim_params} are used, and every reported curve is averaged over
$N_{\mathrm{trials}}$ independent Monte Carlo channel and placement
realizations, where $N_{\mathrm{trials}} = 5000$ for CDF plots and $N_{\mathrm{trials}} = 2000$ for others.

\begin{table}[t]
\centering
\caption{Default Simulation Parameters}
\label{tab:sim_params}
\renewcommand{\arraystretch}{1.2} 
\setlength{\tabcolsep}{4pt} 
\scriptsize
\begin{tabularx}{\columnwidth}{@{} 
    >{\raggedright\arraybackslash}X 
    >{\raggedleft\arraybackslash}p{0.15\columnwidth} | 
    >{\raggedright\arraybackslash}X 
    >{\raggedleft\arraybackslash}p{0.12\columnwidth} @{}}
\toprule
\textbf{Parameter (Symbol)} & \textbf{Value} & \textbf{Parameter (Symbol)} & \textbf{Value} \\
\midrule
Carrier frequency ($f_c$)           & $3.5$\,GHz        & Nominal pilot power ($P_p$)       & $2.5$\,W \\
System bandwidth ($B$)              & $20$\,MHz         & Per-user SINR floor ($\gamma$)    & $5$\,dB \\
Transmit APs ($\Nap$)               & $10$              & Clutter penalty ($\kappa$)        & $0.08$ \\
Antennas per AP ($M$)               & $16$              & Target RCS var.\ ($\sigma_{\mathrm{RCS}}^2$) & $0.5$ \\
Users / targets ($\Nue$ / $\Nt$)    & $4$ / $2$         & Slow-time snapshots ($T$)         & $20$ \\
Radius (AP/UE-tgt)                  & $0.65$ / $1$\,km  & ADMM penalty ($\rho$)             & $1.0$ \\
Per-AP power ($\Pmax$)              & $20$\,W           & ADMM tol.\ ($\epsilon_{\mathrm{pri/dual}}$) & $1.0$ \\
Noise figure ($\mathrm{NF}$)        & $7$\,dB           & Noise temp. ($T_0$)               & $290$\,K  \\
Thermal noise floor ($\sigma_n^2$)  & $-94$\,dBm        & MC trials ($N_{\mathrm{trials}}$) & $\geq 2000$ \\
\bottomrule
\end{tabularx}
\vspace{-3mm} 
\end{table}

\textit{Baselines.} We benchmark 
\cordis-Split (Alg.~\ref{alg:cordis-split}) and \cordis-ADMM
(Alg.~\ref{alg:cordis-admm}) against the following schemes:
\begin{itemize}
\item \textbf{Centralized}: a one-shot solution of the global problem
\eqref{eq:admm-p-global} computed at the CPU with aggregated network CSI and no
consensus splitting. It serves as the performance \emph{upper bound} 
that the distributed algorithms approach.
\item \textbf{LR-MMSE} ($\rho{=}0.5$): the local robust-MMSE precoder of
Section~\ref{subsec:split-distributed-bf} applied with a \emph{fixed}, equal
PSR $\rho_{a_t}{=}0.5$ and no centralized power allocation.
This isolates the value added by the PSR/PA stage of \cordis-Split.
\item \textbf{Global ZF} and \textbf{Global MRT}: network-wide zero-forcing and
maximum-ratio transmit precoders that use the aggregated channel estimates with
uniform power loading and a steering-only (interference-na\"ive) sensing beam.
\end{itemize}

\textit{Metrics.} Communication quality is measured by the per-user SINR
in~\eqref{eq:sinr} and sensing quality by the per-target average post-STAP SCNR
in~\eqref{eq:expedted-scnr}. Because the design enforces a \emph{per-user} QoS
floor, we report the worst-case minimum per-user SINR and the minimum per-target SCNR 
(the mean-based counterparts exhibit the same ordering and trends).
We additionally report the per-user \emph{outage
probability} $P_{\mathrm{out}}(\gamma)=\Pr\{\text{SINR}_u<\gamma\}$ and its
complementary \emph{coverage} $1-P_{\mathrm{out}}$, and, for the
communication-sensing region, the priority-weighted sum SCNR
$\sum_{a_r}\sum_t\omega_t\,\overline{\text{SCNR}}_{a_r,t}$, which is 
the \cordis-ADMM objective in~\eqref{eq:admm-p-global-obj}.

\subsection{Convergence of \cordis-ADMM}
\label{subsec:results-convergence}

Fig.~\ref{fig:convergence} examines the numerical behavior of \cordis-ADMM.
The primal and dual consensus residuals in~\eqref{eq:admm-primal-residual}
and~\eqref{eq:admm-dual-residual} decay monotonically and cross the convergence
tolerance $\epsilon_{\mathrm{pri}}=\epsilon_{\mathrm{dual}}=1$ 
within roughly $150$ iterations for the
displayed realization, confirming that the consensus
splitting of Section~\ref{subsec:admm-consensus-admm} drives the local AP
contributions to network-wide agreement.
We also evaluated the standard
residual-balancing \emph{adaptive penalty} update for $\rho$~\cite{Boyd2011Book};
in our setting it does \emph{not} accelerate convergence. Because the per-user
SOC consensus projection~\eqref{eq:admm-p-central} acts as an inner loop whose
effect on the residuals is not captured by the residual-balancing heuristic,
the adaptive-$\rho$ scheme repeatedly readjusts the penalty without the usual
speed-up. as Fig.~\ref{fig:convergence} shows, its residuals still decrease but
reach the tolerance only at around $250$ iterations. We therefore retain a fixed
penalty $\rho$ in all subsequent experiments.
The center panel tracks the worst-case QoS slack
$\max_u\varepsilon_u$ of the central projection~\eqref{eq:admm-p-central}, which
decreases by several orders of magnitude. Since $\varepsilon_u$
is the only mechanism by which the per-user SINR constraint can be violated, its
decay certifies that the QoS guarantee~\eqref{eq:admm-p-global-sinr} is met at
convergence. The right panel makes this concrete; as the slack shrinks, the
realized worst-user SINR is pulled onto the floor $\gamma=5$\,dB, i.e.,
\cordis-ADMM spends exactly the communication power needed to satisfy the QoS
constraint and reinvests the remainder into sensing. 

\begin{figure}[!t]
\centering
\includegraphics[width=\columnwidth]{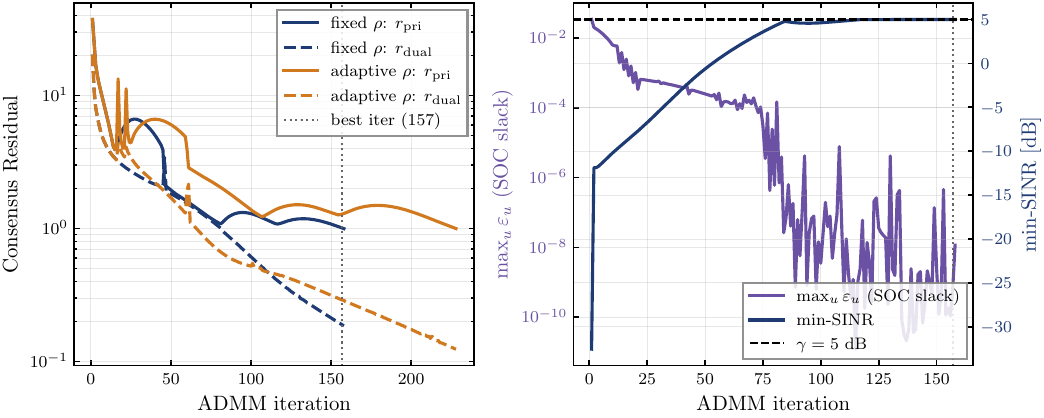}
\caption{Convergence of \cordis-ADMM. (Left) primal and dual consensus
residuals versus ADMM iteration for fixed and adaptive penalty $\rho$; 
(Middle) Maximum per-user SOC slack
$\max_u\varepsilon_u$ (Right) SOC slack vs. realized worst-user SINR.}
\label{fig:convergence}
\end{figure}

\subsection{Communication-Sensing Trade-off and Operating Point}
\label{subsec:results-tradeoff}

\begin{figure}[!t]
\centering
\includegraphics[width=0.9\columnwidth]{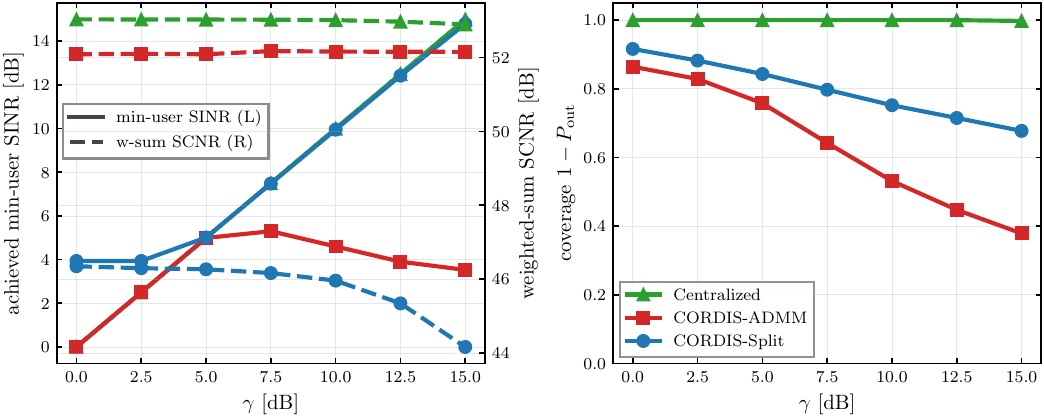}
\caption{Communication-sensing trade-off versus SINR target
$\gamma$. (Left) Achieved worst-case SINR (left axis, solid) and priority-weighted
sum SCNR (right axis, dashed); 
(Right) Communication coverage $1-P_{\mathrm{out}}$ versus $\gamma$.}
\label{fig:cs_tradeoff}
\end{figure}

Fig.~\ref{fig:cs_tradeoff} sweeps the SINR target
$\gamma\in[0,15]$\,dB and exposes the communication-sensing trade-off. 
In the left panel, the Centralized and \cordis-Split approaches 
raise the worst-user SINR in lock-step with the prescribed
floor, while \cordis-ADMM holds the worst-case user close to
$\gamma$ and declines to over-serve it once $\gamma\gtrsim7.5$\,dB. The payoff
is visible on the sensing axis: \cordis-ADMM sustains a weighted-sum SCNR of
$\approx\!52$\,dB, within about $1$\,dB of the centralized ceiling and roughly
$6$\,dB above \cordis-Split, essentially flat in $\gamma$ because its joint
formulation reallocates every available spatial and power degree of freedom toward
the targets. \cordis-Split, constrained by its fixed local beamformers and
heuristic power split, cannot perform this reallocation and its sum SCNR even
declines as more power is diverted to communication.
The right panel shows the complementary cost: communication coverage decreases with
$\gamma$ for both distributed schemes, and the conservative \cordis-Split
retains a higher coverage (e.g., $\approx\!0.84$ versus $\approx\!0.76$ at
$\gamma{=}5$\,dB) precisely because it does not push users onto the floor. The
Centralized approach, with global CSI, maintains unit coverage throughout. These
curves identify $\gamma^\star=5$\,dB as a
balanced operating point, which we adopt for the remaining experiments.

\subsection{Operating-Point Performance Distributions}
\label{subsec:results-cdf}

\begin{figure}[!t]
\centering
\includegraphics[width=\columnwidth]{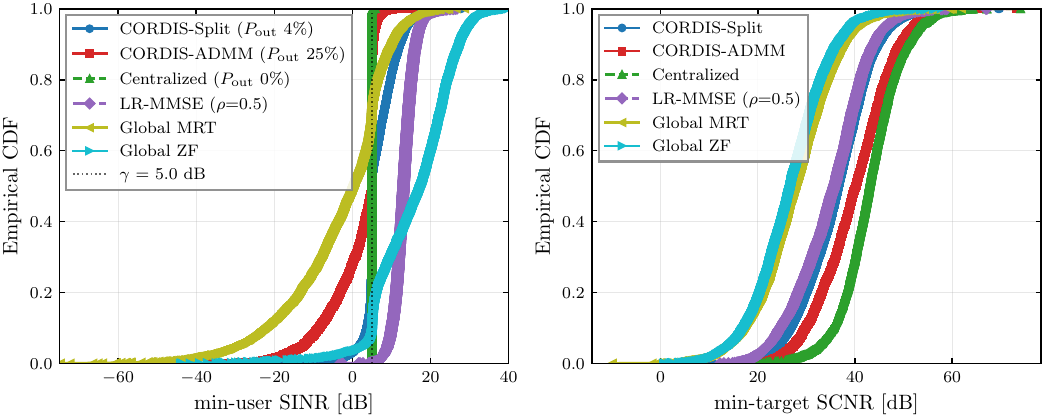}
\caption{Empirical CDFs at the operating point $\gamma^\star=5$\,dB. (Left)
Worst-case SINR, incl.~per-user outage $P_{\mathrm{out}}$; (Right)
Worst-case SCNR.}
\label{fig:cdf}
\end{figure}

Fig.~\ref{fig:cdf} shows empirical CDFs of the worst-case SINR and
SCNR for $\gamma^\star=5$\,dB. On the communication side (left), the
interference-na\"ive and zero-forcing baselines (Global MRT/ZF) and fixed-split
LR-MMSE achieve the highest SINR, since they devote essentially all
resources to communication; \cordis-Split inherits a comfortable margin above
the floor ($P_{\mathrm{out}}\!\approx\!4\%$), whereas \cordis-ADMM drives the
worst-case SINR to $\gamma$ in order to free spatial and
power resources for sensing.
Because this floor is enforced on the estimated channels
$\widehat{\mathbf{h}}_{a_t u}$, the residual CSI error then spreads the realized worst-case SINR around $\gamma$, 
pushing a larger fraction of
users below it ($P_{\mathrm{out}}\!\approx\!25\%$).
The sensing panel (right) reveals what this buys:
\cordis-ADMM's SCNR distribution is statistically indistinguishable from the
centralized one and lies several dB above \cordis-Split, while the
communication-optimal baselines lose roughly $10$\,dB in SCNR. 
The methods that work best on the left are the worst on the right, and
\cordis-ADMM is the only \emph{distributed} approach that matches centralized
sensing performance. \cordis-Split occupies the opposite corner: robust
communication and modest sensing, making the two proposed algorithms complementary
operating modes rather than strict substitutes.

\subsection{Clutter Suppression}
\label{subsec:results-clutter}

\begin{figure}[!t]
\centering
\includegraphics[width=0.8\columnwidth]{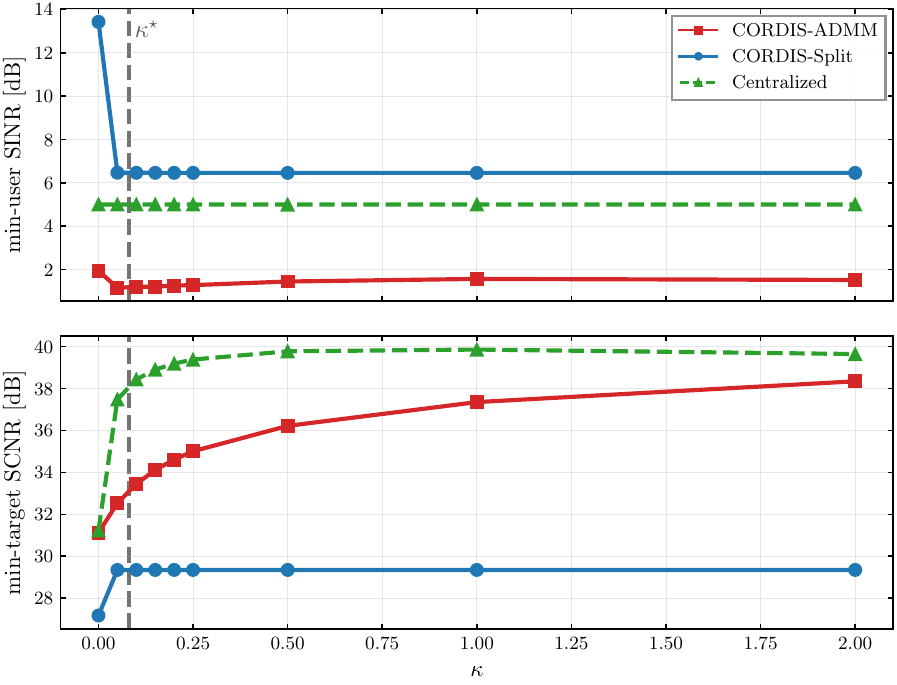}
\caption{Effect of the clutter penalty $\kappa$. (Top) Worst-case SINR;
(Bottom) Worst-case SCNR.}
\label{fig:clutter}
\end{figure}

Fig.~\ref{fig:clutter} isolates the role of the clutter-regularization
parameter $\kappa$ in~\eqref{eq:admm-linear-objective}. As $\kappa$
increases, the penalty
$\kappa\,\tr(\mathbf{W}_{a_t}^H\mathbf{C}_{a_t}\mathbf{W}_{a_t})$ progressively
suppresses energy radiated into the statistical clutter subspace, and
the worst-case SCNR of the joint schemes climbs monotonically before
saturating once the clutter is effectively nulled. \cordis-ADMM improves by
several dB and tracks the centralized profile, while the Centralized bound
itself shows the same diminishing-returns shape. The top panel
confirms that this sensing gain is almost free in terms of communication terms. The
worst-case SINR is insensitive to $\kappa$ since the penalty acts
on the sensing-dominated spatial directions. \cordis-Split is flat in $\kappa$
because its fixed null-space-projected sensing beam cannot react to a CPU-level
penalty, a limitation that motivates the joint design. The value $\kappa^\star=0.08$ 
lies at the knee of the SCNR curve, capturing the bulk
of the suppression benefit with minimal impact on communication. We note that
this behavior requires the angular support of the clutter to be (at least partially)
separable from the target directions; when the two coincide, penalizing clutter
also attenuates the target return and the gain vanishes.

\subsection{Robustness to Imperfect CSI}
\label{subsec:results-csi}

\begin{figure}[!t]
\centering
\includegraphics[width=0.8\columnwidth]{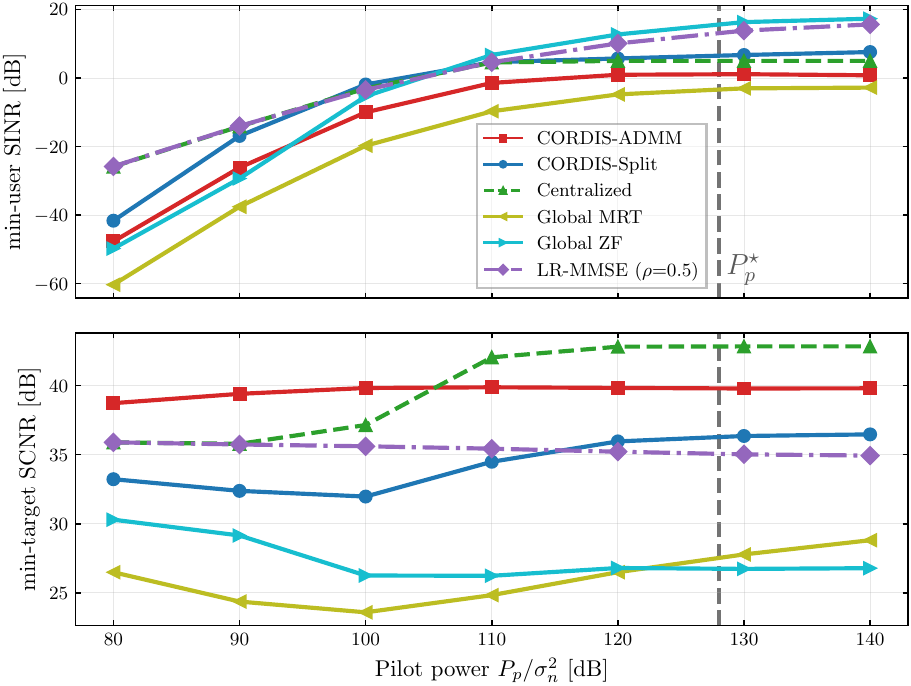}
\caption{Robustness to imperfect CSI vs. the uplink pilot power
$P_p/\sigma_n^2$. (Top) Worst-case SINR; (Bottom) Worst-case SCNR.}
\label{fig:csi}
\end{figure}

Fig.~\ref{fig:csi} sweeps the uplink pilot power $P_p/\sigma_n^2$, which
controls the MMSE channel-estimation quality. 
On the sensing side
(bottom), \cordis-ADMM is strikingly \emph{robust}: its worst-case SCNR remains
within about $1$\,dB of its best value, even when the
estimates are poor, because the STAP sensing metric in~\eqref{eq:expedted-scnr}
depends on the clutter-plus-noise covariance and the target
steering geometry rather than instantaneous UE channel estimates. The
Centralized approach, by contrast, must estimate its way to a comparable SCNR
and only overtakes \cordis-ADMM once the CSI is nearly perfect, while the
communication-oriented baselines perform relatively poorly. For
communication (top), all methods degrade with decreasing pilot power, but the
robust-MMSE-based designs (\cordis-Split, LR-MMSE) and \cordis-ADMM
degrade more gracefully, owing to the explicit incorporation
of the estimation-error covariance $\widetilde{\mathbf{R}}_{a_t u}$ into the
precoder and the SINR model. At $P_p^\star$ (vertical dashed line), the network achieves 
a high-quality semsing-communication balance. 

\subsection{Scalability in the Locally Low-Rank Regime}
\label{subsec:results-lowrank}

Fig.~\ref{fig:lowrank} assumes AP arrays of only 
$M=3$ antennas, with the total number of transmit antennas held
constant by increasing $\Nap$. Once $\Nue>M$, the local channel matrices
become rank-deficient and no AP can null all MUI, which is
the locally low-rank regime that motivated \cordis-ADMM. In the shaded region, 
the consequences for approaches with
fixed local beamformers are severe: the per-user outage of \cordis-Split
climbs past $0.5$ and that of Global MRT toward $0.8$, while their worst-case
SCNR simultaneously erodes. \cordis-ADMM behaves entirely differently. By
coordinating beamformers across APs via the consensus mechanism, the
outage remains low ($\approx\!0.27$ at the heaviest load), and the SCNR remains
within about $1$\,dB of the centralized
bound. This demonstrates the \cordis framework's central claim:
when local channels are well-conditioned, the lightweight \cordis-Split suffices
and even excels in communication robustness, but as the local channels become
rank-deficient -- arguably the operating point of practical dense
cell-free deployments -- only the joint consensus optimization of \cordis-ADMM
preserves both communication and sensing QoS, at a fronthaul cost that is
independent of the number of APs and antenna array size.

\begin{figure}[!t]
\centering
\includegraphics[width=0.8\columnwidth]{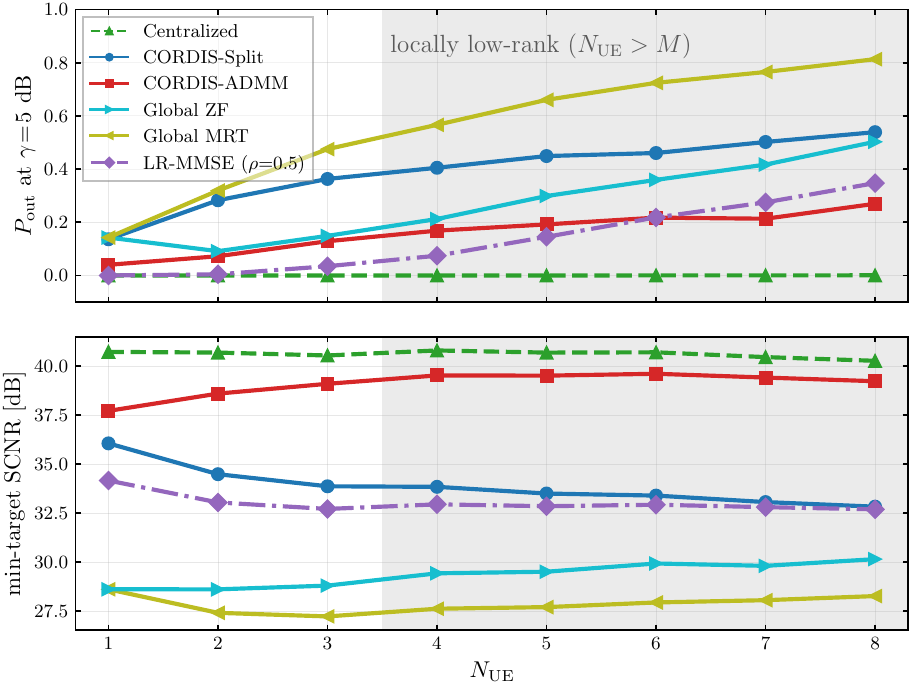}
\caption{Scalability as the user load $\Nue$ grows for a fixed per-AP array of
$M$ antennas. (Top) Per-user outage at $\gamma=5$\,dB; (Bottom) Worst-case
SCNR.}
\label{fig:lowrank}
\end{figure}

\subsection{Fronthaul and Computational Overhead}
\label{subsec:results-overhead}

\begin{table*}[t]
\centering
\caption{Per-AP fronthaul coordination overhead. CORDIS-ADMM runs $T_{\rm ADMM}$ rounds per solve ($T_{\rm ADMM}\approx92.0$ here).}
\label{tab:fronthaul}
\begin{tabular}{lccccc}
\toprule
Algorithm & Data shared & Size & Real/round & Iter. & Scalable \\
\midrule
Centralized & $\widehat{\mathbf{H}}_{a_t},\ \mathbf{W}_{a_t}$ & $\mathbb{C}^{M_t \times \Nue} + \mathbb{C}^{M_t \times |\mathcal{D}|}$ & 320 & 1 & $\times$ \\
CORDIS-Split & $\{\widehat{\beta}_{a_t u}, \widetilde{g}_{a_t u}, e_{a_t u}^{(s)}\}_{u \in \mathcal{U}},\ z_{a_t}, \widetilde{q}_{a_t},\ \rho^\ast_{a_t}$ & $(3 \Nue + 3) \times \mathbb{R}$ & 15 & 1 & $\checkmark$ \\
CORDIS-ADMM & $\mathbf{l}_{a_t u}\ (\mathrm{up}),\ \widetilde{\boldsymbol{\Sigma}}_u\ (\mathrm{dl})$ & $2 \times \Nue \times \mathbb{C}^{|\mathcal{D}|+1}$ per outer iter & 104 & 92 & $\checkmark$ \\
\bottomrule
\end{tabular}
\vspace{-1mm}
\end{table*}

\begin{figure}[!t]
\centering
\includegraphics[width=0.6\columnwidth]{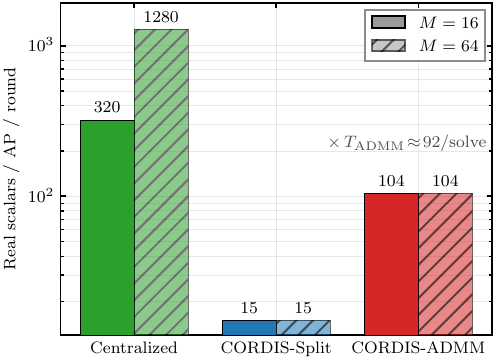}
\caption{Per-AP fronthaul exchange (real scalars per AP per round, log scale)
for two array sizes $M\in\{16,64\}$.}
\label{fig:fronthaul}
\end{figure}

Table~\ref{tab:fronthaul} summarizes the
fronthaul load per round, while Fig.~\ref{fig:fronthaul} quantifies the load 
for two representative array sizes. The centralized approach must
collect the $M$-dimensional channel estimates and precoders at the CPU, so
its per-AP payload grows linearly with $M$, from $320$ real
scalars at $M{=}16$ to $1280$ at $M{=}64$, and, once aggregated, also with the
number of APs $\Nap$. Both \cordis\ schemes, in contrast, exchange only
low-dimensional, $M$-independent summaries: \cordis-Split forwards 
three compressed scalars per user ($15$ scalars in a single round),
while \cordis-ADMM exchanges the contribution vector $\mathbf{l}_{a_t u}$ 
($104$ scalars per round), repeated over
$T_{\mathrm{ADMM}}\!\approx\!92$ consensus rounds in our experiments.
\cordis-ADMM therefore trades a larger cumulative exchange for two
properties the centralized scheme lacks and that are precisely what
matters at scale: every message is independent of the array size $M$, and,
because the heavy spatial optimization remains local, the per-AP computation is
independent of $\Nap$. This scaling is also mirrored in the computational cost. 
Taken together with Figs.~\ref{fig:cdf}--\ref{fig:lowrank}, these
results show that \cordis-ADMM recovers near-centralized communication and
sensing performance, including in the rank-deficient regime where the
lightweight \cordis-Split degrades, while keeping both the fronthaul payload
and the per-AP computation independent of the antenna count and the network
size, the two dimensions along which dense cell-free ISAC deployments actually
grow.


\section{Conclusion}
\label{sec:conclusion}

We have introduced \cordis, a coordinated resource-allocation framework for
distributed cell-free ISAC that makes the communication-sensing trade-off
explicit while respecting fronthaul and local computation limits. Within this framework we
developed two complementary algorithms: \cordis-Split, which pairs fixed local
robust beamformers with a centralized convex power allocation and exchanges only
a few scalars per user, and \cordis-ADMM, which jointly optimizes beamforming and
power allocation through a consensus ADMM procedure whose per-round exchange and
per-AP computation are independent of the array size and number of APs. Over
a 3GPP UMi correlated-Rician channel with clutter and imperfect CSI,
\cordis-ADMM tracked the centralized sensing bound to within about $1$\,dB and
degraded gracefully under CSI estimation error, and in the locally low-rank regime
$\Nue>\Mt$, where the lightweight \cordis-Split saturates, it preserved both the
per-user QoS and the target SCNR. 

\bibliographystyle{IEEEtran}
\bibliography{ref}

@IEEEtranBSTCTL{BSTcontrol,
  CTLuse_forced_etal       = "yes",
  CTLmax_names_forced_etal = "7",
  CTLnames_show_etal       = "1"
}

@ARTICLE{Liu2022dualFunctionalISAC,
  author={Liu, Fan and Cui, Yuanhao and Masouros, Christos and Xu, Jie and Han, Tony Xiao and Eldar, Yonina C. and Buzzi, Stefano},
  journal={IEEE J. Sel. Areas in Commun.}, 
  title={{Integrated Sensing and Communications: Toward Dual-Functional Wireless Networks for 6G and Beyond}}, 
  year={2022},
  volume={40},
  number={6},
  pages={1728-1767},
  doi={10.1109/JSAC.2022.3156632}
}

@INPROCEEDINGS{Wymeersch2021ISAC6G,
  author={Wymeersch, Henk and Shrestha, Deep and de Lima, Carlos Morais and Yajnanarayana, Vijaya and Richerzhagen, Bj{\"o}rnson and Keskin, Musa Furkan and Schindhelm, Kim and Ramirez, Alejandro and Wolfgang, Andreas and de Guzman, Mar Francis and Haneda, Katsuyuki and Svensson, Tommy and Baldemair, Robert and Parkvall, Stefan},
  booktitle={Proc. IEEE Int'l Symp. on Personal, Indoor and Mobile Radio Commun. (PIMRC)}, 
  title={{Integration of Communication and Sensing in 6G: a Joint Industrial and Academic Perspective}}, 
  year={2021},
  volume={},
  number={},
  doi={10.1109/PIMRC50174.2021.9569364}
}

@ARTICLE{Liu2026multiDomainOpt,
  author={Liu, Rang and Li, Ming and Zafari, Mehdi and Ottersten, Björn and Swindlehurst, A.},
  journal={IEEE Wireless Commun. (Early Access)}, 
  title={{Multi-Domain Optimization Framework for ISAC: From Electromagnetic Shaping to Network Cooperation}}, 
  year={2026},
  volume={},
  number={},
  doi={10.1109/MWC.2025.3646980}
}

@ARTICLE{Ngo2017cellfreeVsSmallcell,
  author={Ngo, Hien Quoc and Ashikhmin, Alexei and Yang, Hong and Larsson, Erik G. and Marzetta, Thomas L.},
  journal={IEEE Trans. Wireless Commun.}, 
  title={{Cell-Free Massive MIMO Versus Small Cells}}, 
  year={2017},
  volume={16},
  number={3},
  pages={1834-1850},
  doi={10.1109/TWC.2017.2655515}
}

@article{Interdonato2019ubiquitousCF,
  author  = {Giovanni Interdonato and Emil Bj{\"o}rnson and Hien Quoc Ngo and P{\aa}l Frenger and Erik G. Larsson},
  title   = {{Ubiquitous Cell-Free Massive {MIMO} Communications}},
  journal = {EURASIP J. Wireless Commun. and Networking},
  year    = {2019},
  volume  = {2019},
  number  = {1},
  pages   = {197},
  doi     = {10.1186/s13638-019-1507-0}
}

@article{Demir2021foundationsCF,
    author = {{\"O}zlem Tu{\u{g}}fe Demir and Emil Bj{\"o}rnson and Luca Sanguinetti},
    title = {{Foundations of User-Centric Cell-Free Massive MIMO}},
    journal = {Foundations and Trends in Sig. Proc.},
    volume = {14},
    number = {3-4},
    pages = {162-472},
    year = {2021},
    month = {01},
    issn = {1932-8346},
    doi = {10.1561/2000000109},
    eprint = {https://www.emerald.com/ftsig/article-pdf/14/3-4/162/11046948/2000000109en.pdf},
}

@ARTICLE{Femenias2025scalableISAC_HWI,
  author={Femenias, Guillem and Riera-Palou, Felip},
  journal={IEEE Open J. Commun. Soc.}, 
  title={{ISAC in Scalable Cell-Free Massive MIMO Networks With Hardware Impairments}}, 
  year={2025},
  volume={6},
  number={},
  pages={8793-8815},
  doi={10.1109/OJCOMS.2025.3621773}
}

@inproceedings{Zafari2025confAsilomar,
      author={Mehdi Zafari and Rang Liu and A. Swindlehurst},
  booktitle={Proc. 59th Asilomar Conf. on Sig., Sys., and Comp.}, 
  title={{Coordinated Decentralized Resource Optimization for Cell-Free ISAC Systems}}, 
  year={2025},
  volume={},
  number={},
  pages={912-917}
}

@ARTICLE{Mao2024CSregion,
  author={Mao, Weihao and others},
  journal={IEEE Trans. Wireless Commun.}, 
  title={{Communication-Sensing Region for Cell-Free Massive MIMO ISAC Systems}}, 
  year={2024},
  volume={23},
  number={9},
  pages={12396-12411},
  doi={10.1109/TWC.2024.3392330}
}

@techreport{ITU_R_M2160_2023,
  author       = "{ITU-R}",
  title        = "{Framework and overall objectives of the future development of IMT for 2030 and beyond}",
  institution  = "{International Telecommunication Union, Radiocommunication Sector}",
  type         = "{Recommendation}",
  number       = "{ITU-R M.2160-0}",
  address      = "{Geneva, Switzerland}",
  month        = nov,
  year         = {2023},
}

@techreport{ETSI_GR_ISC_001_V111_2025,
  author       = {{European Telecommunications Standards Institute (ETSI)}},
  title        = {{Integrated Sensing And Communications (ISAC); Use Cases and Deployment Scenarios}},
  institution  = {{ETSI}},
  type         = {{Group Report}},
  number       = {{ETSI GR ISC 001}},
  version      = {{V1.1.1}},
  address      = {{Sophia Antipolis Cedex, France}},
  month        = mar,
  year         = {2025},
}

@techreport{3GPP_TR_22837_V1940_2024,
  author       = {{3rd Generation Partnership Project (3GPP)}},
  title        = {{Study on Integrated Sensing and Communication}},
  institution  = {{3GPP}},
  type         = {{Technical Report}},
  number       = {{TR 22.837}},
  version      = {{V19.4.0}},
  month        = jun,
  year         = {2024},
}

@ARTICLE{Elfiatoure2025multiTargetCF,
  author={Elfiatoure, Mohamed and Mohammadi, Mohammadali and Quoc Ngo, Hien and Shin, Hyundong and Matthaiou, Michail},
  journal={IEEE Trans. Wireless Commun.}, 
  title={{Multiple-Target Detection in Cell-Free Massive MIMO-Assisted ISAC}}, 
  year={2025},
  volume={24},
  number={5},
  pages={4283-4298},
  doi={10.1109/TWC.2025.3542206}
}

@ARTICLE{Meng2025CoopIsacNet,
  author={Meng, Kaitao and Masouros, Christos and Petropulu, Athina P. and Hanzo, Lajos},
  journal={IEEE Trans. Wireless Commun.}, 
  title={{Cooperative ISAC Networks: Performance Analysis, Scaling Laws, and Optimization}}, 
  year={2025},
  volume={24},
  number={2},
  pages={877-892},
  doi={10.1109/TWC.2024.3491356}
}

@ARTICLE{Behdad2024multistatic,
  author={Behdad, Zinat and Demir, {\"O}zlem Tu{\u g}fe and Sung, Ki Won and Bj{\"o}rnson, Emil and Cavdar, Cicek},
  journal={IEEE Trans. Wireless Commun.}, 
  title={{Multi-Static Target Detection and Power Allocation for Integrated Sensing and Communication in Cell-Free Massive MIMO}}, 
  year={2024},
  volume={23},
  number={9},
  pages={11580-11596},
  doi={10.1109/TWC.2024.3383209}
}

@ARTICLE{Huang2022CoordPowerControl,
  author={Huang, Yi and Fang, Yuan and Li, Xinmin and Xu, Jie},
  journal={IEEE Trans. Vehic. Tech.}, 
  title={{Coordinated Power Control for Network Integrated Sensing and Communication}}, 
  year={2022},
  volume={71},
  number={12},
  pages={13361-13365},
  doi={10.1109/TVT.2022.3194139}
}

@ARTICLE{Femenias2025ScalableCFmMIMO,
  author={Femenias, Guillem and Riera-Palou, Felip},
  journal={IEEE Trans. Vehic. Tech.},
  title={{Scalable Cell-Free Massive MIMO-Based Integrated Sensing and Communication}}, 
  year={2025},
  volume={74},
  number={10},
  pages={15643-15659},
  doi={10.1109/TVT.2025.3566710}
  }

@ARTICLE{Demirhan2025CellFreeISAC,
  author={Demirhan, Umut and Alkhateeb, Ahmed},
  journal={IEEE Trans. Commun.}, 
  title={{Cell-Free ISAC MIMO Systems: Joint Sensing and Communication Beamforming}}, 
  year={2025},
  volume={73},
  number={6},
  pages={4454-4468},
  doi={10.1109/TCOMM.2024.3490740}
  }

@ARTICLE{Bjornson2020ScalableCFmMIMO,
  author={Björnson, Emil and Sanguinetti, Luca},
  journal={IEEE Trans. Commun.}, 
  title={{Scalable Cell-Free Massive MIMO Systems}}, 
  year={2020},
  volume={68},
  number={7},
  pages={4247-4261},
  doi={10.1109/TCOMM.2020.2987311}
  }

@ARTICLE{Zou2024DistVsCent,
  author={Zou, Qinglin and Behdad, Zinat and Tu{\u g}fe Demir, {\"O}zlem and Cavdar, Cicek},
  journal={IEEE Wireless Commun. Letters}, 
  title={{Distributed Versus Centralized Sensing in Cell-Free Massive MIMO}}, 
  year={2024},
  volume={13},
  number={12},
  pages={3345-3349},
  doi={10.1109/LWC.2024.3462710}
  }

@book{Richards2010BookRadar,
    author = {Mark A. Richards  and James A. Scheer  and William A. Holm },
    title = {{Principles of Modern Radar: Basic Principles}},
    publisher = {The Institution of Engineering and Technology},
    year = {2010},
    doi = {10.1049/SBRA021E}
}

@inproceedings{Zafari2026ASSENT,
    title={{ASSENT: Learning-Based Association Optimization for Distributed Cell-Free ISAC}}, 
    author={Mehdi Zafari and A. Swindlehurst},
    booktitle={Proc. IEEE Int'l Conf. on Commun. (ICC)}, 
    year={2026}
}

@ARTICLE{Memisoglu2024Scheduling,
  author={Memisoglu, Ebubekir and Janjua, Muhammad Bilal and Arslan, Hüseyin},
  journal={IEEE Wireless Commun. Letters}, 
  title={{Power-Efficient Time-Domain Scheduling for ISAC Beamforming}}, 
  year={2024},
  volume={13},
  number={10},
  pages={2837-2841},
  doi={10.1109/LWC.2024.3448528}
  }

@INPROCEEDINGS{Zafari2024ADMM,
    author={Zafari, Mehdi and Pandey, Divyanshu and Doost-Mohammady, Rahman and Uribe, César A.},
    booktitle = {Proc. 58th Asilomar Conf. Sig., Syst. and Comput.},
    title={{ADMM for Downlink Beamforming in Cell-Free Massive MIMO Systems}}, 
    year={2024},
    volume={},
    number={},
    pages={623-628},
    doi={10.1109/IEEECONF60004.2024.10943106}
}

@article{Boyd2011Book,
    year = {2011},
    volume = {3},
    journal = {Foundations and Trends® in Machine Learning},
    title = {{Distributed Optimization and Statistical Learning via the Alternating Direction Method of Multipliers}},
    doi = {10.1561/2200000016},
    issn = {1935-8237},
    number = {1},
    pages = {1-122},
    author = {Stephen Boyd and Neal Parikh and Eric Chu and Borja Peleato and Jonathan Eckstein}
}

\end{document}